\documentclass[12pt]{elsarticle}

\usepackage[T1]{fontenc}
\usepackage{caption}
\usepackage{times}
\usepackage{color,graphicx}
\usepackage{subfig,array}
\usepackage{amsthm,amssymb,amsmath,tensor}
\usepackage{tikz}
\usetikzlibrary{arrows,decorations.pathreplacing}
\tikzstyle{vertex}=[circle, draw, inner sep=0pt, minimum size=4pt]
\newcommand{\vertex}{\node[vertex]}

    \makeatletter
    \def\ps@pprintTitle{%
       \let\@oddhead\@empty
       \let\@evenhead\@empty
       \def\@oddfoot{\reset@font\hfil\thepage\hfil}
       \let\@evenfoot\@oddfoot
    }
    \makeatother

\newcommand{\bes}{\begin{eqnarray*}}
	\newcommand{\ees}{\end{eqnarray*}}
\newcommand{\bpm}{\begin{pmatrix}}
	\newcommand{\epm}{\end{pmatrix}}

\def\diag{{\rm diag}\,}

\newtheorem{lemma}{Lemma}
\newtheorem{theorem}{Theorem}
\newtheorem{corollary}{Corollary}
\newtheorem{remark}{Remark}
\newtheorem{proposition}{Proposition}

\newtheorem{example}{Example}

\begin{document}

\begin{frontmatter}



\title{Perfect quantum state transfer using Hadamard diagonalizable graphs}

\author[label1,label2]{Nathaniel Johnston}
\address[label1]{Department of Mathematics \& Computer Science, Mount Allison University, Sackville, NB, Canada E4L 1E4}
\address[label2]{Department of Mathematics \& Statistics, University of Guelph, Guelph, ON, Canada N1G 2W1}

\author[label3]{Steve Kirkland}
\address[label3]{Department of Mathematics, University of Manitoba, Winnipeg, MB, Canada  R3T 2N2}

\author[label2,label3,label4]{Sarah Plosker}
\ead{ploskers@brandonu.ca}
\address[label4]{Department of Mathematics \& Computer Science, Brandon University, Brandon,  MB, Canada R7A 6A9}

\author[label4]{Rebecca Storey}

\author[label3]{Xiaohong Zhang}

\begin{abstract}  Quantum state transfer within a quantum computer can be achieved by
using a network of qubits, and such a network can be modelled
mathematically by a graph. Here, we focus on the corresponding Laplacian
matrix, and those graphs for which the Laplacian can be diagonalized by
a Hadamard matrix. We give a simple eigenvalue characterization for when
such a graph has perfect state transfer at time $\pi /2$; this
characterization allows one to choose the correct eigenvalues to build
graphs having perfect state transfer. We characterize the graphs that
are diagonalizable by the standard Hadamard matrix, showing a direct
relationship to cubelike graphs. We then give a number of constructions
producing a wide variety of new graphs that exhibit perfect state
transfer, and we consider several corollaries in the settings of both
weighted and unweighted graphs, as well as how our results relate to the
notion of pretty good state transfer. Finally, we give an optimality
result, showing that among regular graphs of degree at most $4$, the
hypercube is the sparsest Hadamard diagonalizable connected unweighted
graph with perfect state transfer.
\end{abstract}


\end{frontmatter}




\section{Introduction}

Accurate transmission of quantum states between processors and/or
registers of a quantum computer is critical for short distance
communication in a physical quantum computing scheme. Bose
\cite{Bose} first proposed the use of spin chains to accomplish this
task over a decade ago. Since then, much work has been done on
\emph{perfect state transfer (PST)}, which accomplishes this task
perfectly in the sense that the state read out by the receiver at some
time $t_{0}$ is, with probability equal to one, identical up to complex
modulus to the input state of the sender at time $t=0$.

Many families of graphs have been found to exhibit PST, including the
join of a weighted two-vertex graph with any regular graph
\cite{Ang09}, Hamming graphs \cite{Ang09} (see also
\cite{BGS08,CDDEKL,CDEL}), a family of double-cone non-periodic graphs
\cite{Ang10}, and a family of integral circulant graphs
\cite{B13} (see also \cite{Ang10}). It is easy to see that the
Cartesian product of two graphs having PST at the same time also has PST
\cite[Sec.~3.3]{AlEtAl}. Much work has also been done with respect
to analyzing the sensitivity
\cite{Andloc,numerical1,GKLPZ,Kay06,Steve2015,numerical2,numerical3},
or even correcting errors \cite{Kay15}, of quantum spin systems.
Signed graphs and graphs with arbitrary edge weights have also been
considered (see \cite{signed} and the references therein), due to
the intriguing fact that certain graphs that do not exhibit PST when
unsigned/unweighted can exhibit PST when signed or weighted properly.
Three articles particularly relevant to the work herein are
\cite{BGS08,cubelike}, which characterized perfect state transfer in
cubelike graphs, a family of graphs that are Hadamard diagonalizable,
and \cite{doublecover}, which shows that perfect state transfer
occurs in graphs constructed in a manner similar to our merge operation,
called the ``$\ltimes $'' operation.

The general approach taken in the literature is to model a quantum spin
system with an undirected connected graph, where the dynamics of the
system are governed by the Hamiltonian of the system: for $XX$ dynamics
the Hamiltonian is the adjacency matrix corresponding to the graph, and
for Heisenberg ($XXX$) dynamics the Hamiltonian is the Laplacian matrix
corresponding to the graph. In the case of $XXX$ dynamics (on which we
focus exclusively in this paper), there is more structure to work with
since we know that the smallest eigenvalue of a Laplacian matrix is
zero, with corresponding eigenvector $\mathbf{1}$ (the all-ones vector).
We note in passing that in the case of regular graphs (which we deal
with frequently in this paper), presence or absence of PST is identical
under $XX$ dynamics and $XXX$ dynamics.

Our contribution to the theory of perfect state transfer is to
characterize graphs that are diagonalizable by the standard Hadamard
matrix, connecting this property with the notion of cubelike graphs, and
to detail procedures for creating new graphs with PST.

Our focus on graphs having a Hadamard diagonalizable Laplacian is not
as restrictive as it might seem at first glance; Hadamard matrices are
ubiquitous in quantum information theory, and because of the special
structure of Hadamard matrices the corresponding graphs tend to exhibit
a good deal of symmetry. As a result, many of the known graphs with PST
are actually Hadamard diagonalizable, such as the hypercube.
Furthermore, integer-weighted graphs with Hadamard diagonalizable
Laplacian are convenient to work with in our setting because they are
known to be regular, with spectra consisting of even integers (see
\cite{ShaunSteve} and Theorem~\ref{thm:SS} below); consequently the
corresponding graph often exhibits PST between two of its vertices at
time $t_{0}=\pi /2$ (see Theorem~\ref{thm:eigenPST} for a more specific
statement).

In Section~\ref{graph}, we give a quick review of the graph theory and
quantum state transfer definitions and tools that we will use. In
Section~\ref{pst_basic_results}, we give an eigenvalue characterization
connecting a graph being Hadamard diagonalizable and it having PST at
time $\pi /2$ between two of its vertices. We further give a connection
between diagonalizability by the standard Hadamard matrix and cubelike
graphs, completely characterizing such graphs. In Section~\ref{pst}, we
describe several ways to construct new Hadamard diagonalizable graphs
from old ones, including our ``merge'' operation, a weighted variant of
the ``$\ltimes $'' operation, which takes two Hadamard diagonalizable
graphs as input, and produces a new (larger) Hadamard diagonalizable
graph with PST as output under a wide variety of conditions. We also
present several results demonstrating the usefulness of this operation
and the types of graphs with PST that it can produce. In
Section~\ref{sec:noninteger}, we discuss how our results generalize to
graphs with non-integer edge weights, which involves the notion of
pretty good state transfer (PGST), and we close in Section~\ref{optim}
with some results concerning the optimality in terms of timing errors
and manufacturing errors of Hadamard diagonalizable graphs.

\section{Preliminaries}\label{graph}

\subsection{Graph Theory Basics}\label{sec:graph_theory_basics}
For a weighted undirected graph $G$ on $n$ vertices, its corresponding
$n\times n$ adjacency matrix $A=(a_{jk})$ is defined by
\begin{eqnarray*}
a_{jk}=
\begin{cases}
w_{j,k}
& \quad \text{if } j \text{ and } k\text{ are adjacent}
\\
0
& \quad \text{otherwise,}
\end{cases}
\end{eqnarray*}
where $w_{j,k}$ is the \emph{weight} of the edge between vertices $j$ and
$k$. Its corresponding $n\times n$ Laplacian matrix is defined by
$L=D-A$, where $D$ is the diagonal matrix of row sums of $A$, known as
the \emph{degree matrix} associated to $G$. Often $w_{j,k}$ in the above
is taken to be 1 for all adjacent $j,k$, in which case the graph is said
to be \emph{unweighted}. A \emph{signed graph} is a graph for which the
non-zero weights can be either $\pm 1$. A \emph{weighted graph} is a
graph for which there is no restriction on $w_{j,k}$ (although the
weights are typically taken to be in $\mathbb{R}$, as they are in this
paper).

An unweighted graph $G$ is \emph{regular} if each of its vertices has
the same number of neighbours, or, more specifically,
\emph{$k$-regular} if each of its vertices has exactly $k$ adjacent
neighbours. The weighted analogue of a regular graph is a graph where
the sum of all the weights of edges incident with a particular vertex
is the same for all vertices. We will be interested in weighted graphs
with this equal ``weighted degree'' property; for simplicity, we simply
call this the \emph{degree} of the graph. A graph is \emph{connected}
if there is a path (a sequence of edges connecting a sequence of
vertices) between every pair of distinct vertices and \emph{complete}
if there is an edge between every pair of distinct vertices (the complete
graph on $n$ vertices is denoted $K_{n}$).

There are several different operations that can be performed to turn two
graphs into a new (typically larger) graph. Specifically, given graphs
$G_{1}=(V_{1}, E_{1})$ and $G_{2}=(V_{2}, E_{2})$, where $V_{1}$ and $E_1$ are the set of all vertices and the set of all edges, respectively, in the graph $G_1$ (and similarly for $V_2$ and $E_2$), then
\vspace*{-2pt}
\begin{enumerate}
\item[1.] The \emph{union} of $G_{1}$ and $G_{2}$ is the graph $G_{1}+G_{2}=(V
_{1}\cup V_{2}, E_{1}\cup E_{2})$;
\item[2.] The \emph{join} of $G_{1}$ and $G_{2}$ is the graph $G_{1}\vee G_{2}=(G
_{1}^{c}+G_{2}^{c})^{c}$ where every vertex of $G_{1}$ is connected to
every vertex of $G_{2}$, and all of the original edges of $G_{1}$ and
$G_{2}$ are retained as well;
\item[3.] The \emph{Cartesian product} of $G_{1}$ and $G_{2}$ is the graph
$G_{1} \square G_{2}=(V_{1}\times V_{2}, E_{3})$ where $V_{1}\times V
_{2}$ is the cartesian product of the two original sets of vertices, and
there is an edge in $G_{1}\square G_{2}$ between vertices $(g_{1}, g
_{2})$ and $(h_{1}, h_{2})$ if and only if either (i) $g_{1}=h_{1}$ and
there is an edge between $g_{2}$ and $h_{2}$ in $G_{2}$, or (ii)
$g_{2}=h_{2}$ and there is an edge between $g_{1}$ and $h_{1}$ in
$G_{1}$.

One can also define the \emph{Cartesian product of weighted graphs}
$G_{1}$ and $G_{2}$ by defining (i) the weight of the edges between
$(g_{1},g_{2})$ and $(g_{1},h_{2})$ in $G_{1}\square G_{2}$ to be the
same as the weight between $g_{2}$ and $h_{2}$ in $G_{2}$, and (ii) the
weight of the edges between $(g_{1},g_{2})$ and $(h_{1},g_{2})$ in
$G_{1}\square G_{2}$ to be the same as the weight between $g_{1}$ and
$h_{1}$ in $G_{1}$; and
\item[4.] If $V_{1}=V_{2}$, let $G_{1}\ltimes G_{2}$ be the graph defined by the\vspace{1.4pt}
adjacency matrix $A(G_{1}\ltimes G_{2})=
\begin{bmatrix}
A(G_{1})&A(G_{2})
\\
A(G_{2})&A(G_{1})
\end{bmatrix}
$,\vspace{1.4pt} where $A(\cdot )$ is the adjacency matrix of the given graph. If the
edge sets of $G_{1}$ and $G_{2}$ are disjoint, then $G_{1}\ltimes G
_{2}$ is a \emph{double cover} of the graph with adjacency matrix
$A(G_{1})+A(G_{2})$.
\vspace*{-2pt}
\end{enumerate}

We recall that a \emph{Hadamard matrix} (or simply, a Hadamard) of order
$n$ is an $n\times n$ matrix $H$ with entries $+1$ and $-1$, such that
$HH^{T}=nI$. Let $H_{1}=
\displaystyle
\begin{bmatrix}
1&1
\\
1&-1
\end{bmatrix}
$, $H_{2}=
\begin{bmatrix}
H_{1}&H_{1}
\\
H_{1}&-H_{1}
\end{bmatrix}
, \dots , H_{n} =
\begin{bmatrix}
H_{n-1} & H_{n-1}
\\
H_{n-1} & -H_{n-1}
\end{bmatrix}
$.\vspace{1.4pt} This construction gives the \emph{standard} Hadamards of order
$2^{n}$. The results herein may be of use in the physical setting
because Hadamards are among the simplest non-trivial gates to implement
in the lab (the standard $n$-qubit Hadamard with a scaling factor of
$1/2^{n/2}$ is frequently used in quantum information theory). From the
definition of a Hadamard matrix, it is clear that any two rows of
$H$ are orthogonal, and any two columns of $H$ are also orthogonal. This
property does not change if we permute rows or columns or if we multiply
some rows or columns by $-1$. This leads to the simple but important
observation that, given a Hadamard matrix, it is always possible to
permute and sign its rows and columns so that all entries of the first
row and all entries of the first column are all 1's. A Hadamard matrix
in this form is said to be \emph{normalized} \cite{ShaunSteve}.
Given a graph $G$ on $n$ vertices with corresponding Laplacian matrix
$L$, if we can write $L=\frac{1}{n}H\Lambda H^{T}$ for some Hadamard
$H$ and diagonal matrix $\Lambda $, then we say that $G$ (or, that
$L$) is \emph{Hadamard diagonalizable}. If $G$ is Hadamard
diagonalizable by some Hadamard $H$, then $G$ is also Hadamard
diagonalizable by a corresponding normalized Hadamard
\cite[Lemma~4]{ShaunSteve}. Thus, there is no loss of generality in
assuming that a Hadamard diagonalizable graph is in fact diagonalized
by a normalized Hadamard matrix. Note that ``normalized'' in this
setting does not imply scaling $H$ to satisfy $\|H\|=1$.

\subsection{Perfect state transfer basics}\label{sec:pst_basics}
A graph exhibits \emph{perfect state transfer (PST)} at time
$t_{0}$ if $p(t_{0}) := |\mathbf{e}_{j}^{T} e^{it_{0}\mathcal{H}}
\mathbf{e}_{k} |^{2}=1$ for some vertices $j\neq k$ and some time
$t_{0}>0$, where $\mathcal{H}$ is the Hamiltonian of the system (either
the adjacency matrix $A$ or the Laplacian matrix $L$, depending on the
system's dynamics). In other words, the graph has perfect state transfer
if and only if $e^{it_{0}\mathcal{H}} \mathbf{e}_{k}$ is a scalar
multiple of $\mathbf{e}_{j}$ (or, equivalently, if $e^{it_{0}
\mathcal{H}} \mathbf{e}_{j}$ is a scalar multiple of $\mathbf{e}_{k}$).
Typically we say that a graph has PST from vertex $j$ to vertex $k$ if
it exhibits PST for some vertices $j$ and $k$ and $j<k$.

A slightly weaker property is that of \emph{pretty good state transfer
(PGST)}: a graph exhibits PGST (for some vertices $j\neq k$) if for
every $\varepsilon >0$, there exists a time $t_{\varepsilon }$ such that
$p(t_{\varepsilon }):= |\mathbf{e}_{j}^{T} e^{it_{\varepsilon }
\mathcal{H}} \mathbf{e}_{k} |^{2} \geq 1-\varepsilon $.

The following observation is well-known.

\begin{remark}\label{rm:gcd}For a general integer-weighted graph $G$, assume that $a$ is the
greatest common divisor of all the edge weights of $G$ and that $L$ is
the Laplacian matrix of $G$. Let $G'$ denote the integer-weighted graph
with Laplacian $1/a\,L$. Since $e^{itL} = e^{ita(\frac{1}{a}L)}$ for all
$t$, we find that $G$ has PST at $\pi /(2a)$ if and only if $G'$ has PST
at $\pi /2$. This allows us to identify more graphs having PST: for
example, if $G$ has PST at $\pi /2$, and we are given the graph with
Laplacian matrix $2L$, we know that it has PST at $\pi /4$.
\end{remark}

\subsection{Cubelike graphs}\label{sec:cubelike}
A large family of graphs, of which the hypercube is a member, is the
family of cubelike graphs \cite{BGS08,cubelike}: Take a set
$C\subset \mathbb{Z}_{2}^{d}=\mathbb{Z}_{2}\times \cdots \times
\mathbb{Z}_{2}$ ($d$ times), where $C$ does not contain the all-zeros
vector. Construct the \emph{cubelike graph} $G(C)$ with vertex set
$V=\mathbb{Z}_{2}^{d}$ and two elements of $V$ are adjacent if and only
if their difference is in $C$. The set $C$ is called the
\emph{connection set} of the graph $G(C)$. The following result
characterizes PST at $\pi /2$ for cubelike graphs.

\begin{theorem}
\label{cubelikePST}\cite[Theorem~1]{BGS08}, \cite[Theorem~2.3]{cubelike} Let
$C$ be a subset of $\mathbb{Z}_{2}^{d}$ and let $\sigma $ be the sum of
the elements of $C$. If $\sigma \neq 0$, then PST occurs in
$G(C)$ from $j$ to $j+\sigma $ at time $\pi /2$. If $\sigma =0$, then
$G(C)$ is \emph{periodic} with period $\pi /2$ (every vertex has perfect
state transfer with itself at time $t_{0}=\pi /2$).
\end{theorem}

The \emph{code} of $G(C)$ is the row space of the $d\times |C|$ matrix
$M$ constructed by taking the elements of $C$ as its columns. When the
sum of the elements of $C$ is zero, it has been shown
\cite{cubelike} that if perfect state transfer occurs on a cubelike
graph, then it must take place at time $\pi /2D$, where $D$ is the
greatest common divisor of the (Hamming) weights of the binary strings
in the code.

\section{Hadamard diagonalizable graphs with PST}\label{pst_basic_results}
The following theorem originally appeared in \cite{ShaunSteve},
restricted to the case of unweighted graphs. The version below allows
for arbitrary integer edge weights. Although its proof is almost
identical to its unweighted version, we include it here for
completeness.

\begin{theorem}
\label{thm:SS}\cite[Theorem~5]{ShaunSteve} If $G$ is an integer-weighted graph that is Hadamard diagonalizable,
then $G$ is regular and all the eigenvalues of its Laplacian are even
integers.
\end{theorem}

\begin{proof}
Without loss of generality we assume that the Laplacian matrix for
$G$ is diagonalized by a normalized Hadamard matrix; observe then that
the first column of that Hadamard is the all-ones vector, and that it
corresponds to the eigenvalue $0$. Choose a non-zero eigenvalue
$\lambda $ of the Laplacian matrix $L$ associated to $G$; the
corresponding column of the Hadamard matrix that diagonalizes $L$ is an
eigenvector corresponding to~$\lambda $. One can split the graph $G$
into two subgraphs, $G_{1}$ and $G_{2}$ (with Laplacians $L_{1}$ and~$L_{2}$), corresponding to the $n/2$ entries of 1 and the $n/2$ entries
of {$-1$} of the eigenvector corresponding to $\lambda $. By applying a
permutation similarity if necessary, we find that
\begin{eqnarray*}
\begin{bmatrix}
L_{1}+X_{1}&-R
\\
-R^{T} & L_{2}+X_{2}
\end{bmatrix}
\begin{bmatrix}
\mathbf{1}
\\
-\mathbf{1}
\end{bmatrix}
=\lambda
\begin{bmatrix}
\mathbf{1}
\\
-\mathbf{1}
\end{bmatrix}
,
\end{eqnarray*}
for some matrices $X_{1}, X_{2}$, and $R$. Necessarily $X_{1}, X_{2}$
are diagonal, and note that we have $X_{1}\mathbf{1}=R\mathbf{1}$ and
$ X_{2}\mathbf{1}=R^{T}\mathbf{1}$.

Since $\lambda \mathbf{1}=L_{1}\mathbf{1}+X_{1}\mathbf{1} +R
\mathbf{1} =2X_{1}\mathbf{1}$, and since $G$ is integer-weighted, we
deduce that $\lambda $ is an even integer. Hence each eigenvalue of the
Laplacian is an even integer.

Next we show that $G$ is regular. For concreteness, suppose that $G$ has
$n$ vertices and that $H$ is a normalized Hadamard matrix so that
$LH=HD$ for some diagonal matrix $D$. Fix an index $j$ between $1$ and
$n$, and let $S_{j}$ be the diagonal matrix with diagonal entries
$\pm 1$ such that $e_{j}^{T}HS_{j}= \mathbf{1}^{T}$. Observe that
$LHS_{j}=HS_{j}D$, and that $HS_{j}$ is also a Hadamard matrix. Since
the $j$-th row of $HS_{j}$ is the all-ones vector and the remaining
rows are orthogonal to it, we deduce that $HS_{j}\mathbf{1}=ne_{j}$.
Consequently, $e_{j}^{T}LHS_{j}\mathbf{1}=ne_{j}^{T}Le_{j}$. On the one
hand, we have $e_{j}^{T}LHS_{j}\mathbf{1}= e_{j}^{T} HS_{j} D
\mathbf{1}=\mathbf{1}^{T} D \mathbf{1}$. Thus, for each $j=1, \ldots
, n, e_{j}^{T}Le_{j} = \frac{1}{n}\mathbf{1}^{T} D \mathbf{1}$, so
$G$ is regular, as desired.
\end{proof}

For an integer-weighted graph that is diagonalizable by some Hadamard
matrix, we now give a precise characterization of its eigenvalues when
it exhibits PST at time $t_{0}=\pi /2$. The proof applies a standard
characterization of PST; see \cite{Kay11}, for example.

\begin{theorem}
\label{thm:eigenPST}
Let $G$ be an integer-weighted graph that is Hadamard diagonalizable by
a Hadamard of order $n$. Let $H=(h_{uv})$ be a corresponding normalized
Hadamard. Denote the eigenvalues of the Laplacian matrix $L$
corresponding to $G$ by $\lambda_{1}, \cdots , \lambda_{n}$, so that
$LH\mathbf{e}_{j}=\lambda_{j}H\mathbf{e}_{j}$, $j=1, \ldots , n$. Then
$G$ has PST from vertex $j$ to vertex $k$ at time $t_{0}=\pi /2$ if and
only if for each $\ell =1,\cdots ,n$, $\lambda_{\ell }\equiv 1-h_{j
\ell }h_{k\ell } \mod {4}$.
\end{theorem}

\begin{proof}
Let $\Lambda $ be the diagonal matrix of eigenvalues such that
$L=\frac{1}{n}H\Lambda H^{T}$, and hence $e^{i(\pi /2)L}=\frac{1}{n}He
^{i(\pi /2)\Lambda }H^{T}$. By the definition of PST, it follows that
$G$ has PST from vertex $j$ to vertex $k$ at $t_{0}=\pi /2$ if and only
if $e^{i(\pi /2)\Lambda }H^{T}\mathbf{e}_{j}$ is a scalar multiple of
$H^{T}\mathbf{e}_{k}$. Since the first column of $H$ is the all ones
vector $\mathbf{1}$, i.e.\ an eigenvector of $L$ corresponding to the
eigenvalue 0, we know that the first entry of $e^{i(\pi /2)\Lambda }H
^{T}\mathbf{e}_{j}$ is $h_{j1}=1$, and the first entry of $H^{T}
\mathbf{e}_{k}=h_{k1}=1$. Thus we deduce that not only is $e^{i(
\pi /2)\Lambda }H^{T}\mathbf{e}_{j}$ a scalar multiple of $H^{T}
\mathbf{e}_{k}$, but that the multiple must be $1$, i.e., we have PST
from vertex $j$ to $k$ at $\pi /2$ if and only if
%
\begin{eqnarray}
\label{pstrows}
e^{i(\pi /2)\Lambda }H^{T}\mathbf{e}_{j}=H^{T}\mathbf{e}_{k}.
\end{eqnarray}
Note that
\begin{eqnarray*}
e^{i(\pi /2)\lambda_{\ell }}=
\begin{cases}
1
& \quad \text{if } \lambda_{\ell }\equiv 0\mod {4}
\\
-1
& \quad \text{if } \lambda_{\ell }\equiv 2\mod {4}.
\end{cases}
\end{eqnarray*}
Consequently, \eqref{pstrows} holds if and only if, for each
$\ell =1,\cdots , n$, if $h_{j\ell }h_{k\ell }=1$ then $\lambda_{
\ell }\equiv 0\mod {4}$, and if $h_{j\ell }h_{k\ell }=-1$ then
$\lambda_{\ell }\equiv 2\mod {4}$. The conclusion follows.
\end{proof}

It is worth noting that {Theorem~\ref{thm:eigenPST}} already gives an
extremely easy method for creating weighted Hadamard diagonalizable
graphs exhibiting PST, since for any normalized Hadamard matrix $H$ we
can choose the eigenvalues in $\Lambda $ to satisfy the required mod~$4$ equation, and then $L = \frac{1}{n} H\Lambda H^{T}$ will necessarily
be the Laplacian of some rational-weighted graph with PST at time
$t_{0} = \pi /2$ (the graph will be integer-weighted provided $n$
divides each edge weight in this construction).

It is known that the adjacency matrix of any cubelike graph is
diagonalized by the standard Hadamard matrix (see \cite{BGS08}).
The following result provides the converse; in the proof, it will be
convenient to denote the graph (possibly containing loops) with
adjacency matrix $A$ by $\Gamma (A)$.

\begin{lemma}
\label{cube-had}
Suppose that $k\in \mathbb{N}$ and that $A$ is a symmetric $(0, 1)$
matrix that is diagonalizable by the standard Hadamard matrix of order
$2^{k}$. Then
\begin{enumerate}
\item[1.] $A$ has constant diagonal;
\item[2.] if $A$ has zero diagonal then it is the adjacency matrix of a cubelike
graph;
\item[3.] if $A$ has all ones on the diagonal, then $A-I$ is the adjacency matrix
of a cubelike graph.
\end{enumerate}
\end{lemma}
\begin{proof}
We proceed by induction on $k$. For $k=1$, it is straightforward to see
that the $(0,1)$ symmetric matrices that are diagonalized by
$H_{1}=
\begin{bmatrix}
1&1
\\
1&-1
\end{bmatrix}
$ are: $
\begin{bmatrix}
0&0
\\
0&0
\end{bmatrix}
$, $
\begin{bmatrix}
1&0
\\
0&1
\end{bmatrix}
$, $
\begin{bmatrix}
0&1
\\
1&0
\end{bmatrix}
$, $
\begin{bmatrix}
1&1
\\
1&1
\end{bmatrix}
$. For these matrices, conclusions (1)--(3) follow readily.

Suppose that the result holds for some $k\in \mathbb{N}$ and that
$A$ is of order $2^{k+1}$. Write the standard Hadamard matrix of order
$2^{k+1}$ as $H_{k+1}=
\begin{bmatrix}
H_{k}&H_{k}
\\
H_{k}&-H_{k}
\end{bmatrix}
$, where $H_{k}$ is the standard Hadamard matrix of order $2^{k}$.
Partition $A$ accordingly as $
\begin{bmatrix}
A_{1}&X
\\
\noalign{\vspace{2pt}}
X^{T}&A_{2}
\end{bmatrix}
$. Then there are diagonal matrices $D_{1},D_{2}$ such that
\begin{equation*}
\begin{bmatrix}
H_{k}&H_{k}
\\
H_{k}&-H_{k}
\end{bmatrix}
\begin{bmatrix}
A_{1}&X
\\ \noalign{\vspace{2pt}}
X^{T}&A_{2}
\end{bmatrix}
\begin{bmatrix}
H_{k}&H_{k}
\\
H_{k}&-H_{k}
\end{bmatrix}
=
\begin{bmatrix}
D_{1}&O
\\
O&D_{2}
\end{bmatrix}
.
\end{equation*}
Hence
$
\begin{bmatrix}
H_{k}(A_{1}+A_{2}+X+X^{T})H_{k}&H_{k}(A_{1}-A_{2}-X+X^{T})H_{k}
\\ \noalign{\vspace{2pt}}
H_{k}(A_{1}-A_{2}+X-X^{T})H_{k}&H_{k}(A_{1}+A_{2}-X-X^{T})H_{k}
\end{bmatrix}
=
\begin{bmatrix}
D_{1}&O
\\
O&D_{2}
\end{bmatrix}
$.\vspace{1.4pt} We deduce that $A_{1}-A_{2}=X-X^{T}$; since $A_{1}-A_{2}$ is
symmetric and $X-X^{T}$ is skew-symmetric, it must be the case that
$A_{1}=A_{2}$ and $X=X^{T}$. Then $H_{k}$ diagonalizes both
$2(A_{1}+X)$ and $2(A_{1}-X)$, and we conclude that $H_{k}$ diagonalizes
$A_{1}$ and diagonalizes~$X$. In particular the induction hypothesis
applies to $A_{1}$ and $X$. Thus $A_{1}$ has constant diagonal, and
hence so does $A$.

Suppose that $A$ has zero diagonal. Applying the induction hypothesis
to $A_{1}$, we find that $\Gamma (A_{1})$ is cubelike. Let $C_{1}$
denote its connection set. Applying the induction hypothesis to $X$,
then either $X$ has zero diagonal and so that $\Gamma (X)$ is a cubelike
graph with connection set $C_{2}$, say, or $\Gamma (X-I)$ is a cubelike
graph with connection set $\tilde{C_{2}}$. Set $C_{2}=\tilde{C_{2}}
\cup \{0\}$.

We label the vertices of the graph $\Gamma (A)$ with vectors in
$\mathbb{Z}_{2}^{k+1}$ in increasing order if considered as binary
numbers. So the first $2^{k}$ rows/columns of $A$ are labelled as
$
\begin{bmatrix}
0
\\
z
\end{bmatrix}
$, where $z\in \mathbb{Z}_{2}^{k}$, and the last $2^{k}$ rows/columns
of $A$ are of\vspace{1.4pt} labelled as $
\begin{bmatrix}
1
\\
z
\end{bmatrix}
$, where $z\in \mathbb{Z}_{2}^{k}$. Now construct the following
connection set:\vspace{1.4pt} $C=\left\{
\begin{bmatrix}
0
\\
x
\end{bmatrix}
, x\in C_{1}\right\} \cup \left\{
\begin{bmatrix}
1
\\
y
\end{bmatrix}
, y\in C_{2}\right\} $. It follows that $A$ is the adjacency matrix of
the cubelike graph with connection set $C$.
\eject

If $A$ has all ones on the diagonal we proceed as above with $A-I$.

This establishes the induction steps for (1)--(3).
\end{proof}

\begin{corollary}
Let $G$ be an unweighted graph with Laplacian matrix $L$. Then $L$ is
diagonalized by the standard Hadamard matrix if and only if $G$ is a
cubelike graph.
\end{corollary}
\begin{proof}
If $L$ is diagonalized by the standard Hadamard matrix, then in
particular $G$ is regular by Theorem~\ref{thm:SS}. Hence the adjacency
matrix of $G$ is diagonalized by the standard Hadamard matrix, so by
Lemma~\ref{cube-had}, $G$ is cubelike. Conversely, if $G$ is cubelike,
it is regular and its adjacency matrix is diagonalized by the standard
Hadamard matrix. We now deduce that $L$ is diagonalized by the standard
Hadamard matrix.
\end{proof}

\section{Creation of new Hadamard diagonalizable graphs with PST}\label{pst}
It is known that the union of a PST graph with itself still exhibits
PST. Here, we show that for a graph $G$ on $n\geq 4$ vertices that is
diagonalizable by some Hadamard matrix and that has PST at time
$\pi /2$, both its complement and the join of $G$ with itself are
Hadamard diagonalizable and have PST at time $t_{0} = \pi /2$.

\begin{proposition}
\label{compljoin}
Let $G$ be an integer-weighted graph on $n\geq 4$ vertices that is
diagonalizable by a Hadamard matrix $H$, and that has perfect state
transfer from vertex $j$ to vertex $k$ at time $t_{0} = \pi /2$. Then
its complement $G^{c}$ is also diagonalizable by $H$, and has the same
PST pairs and PST time as $G$. Furthermore, the join $G\vee G$ of
$G$ with itself is diagonalizable by the Hadamard $
\begin{bmatrix}
H & H
\\
H & -H
\end{bmatrix}
$, and has PST from vertex $j$ to vertex $k$ at time $t_{0} = \pi /2$.
\end{proposition}
\begin{proof}
Without loss of generality we can assume that $H$ is a normalized
Hadamard matrix. The result that $G^{c}$ and $G\vee G$ are
diagonalizable follows from Lemma~7 in~\cite{ShaunSteve}. If we
denote the eigenvalues of the Laplacian of $G$ by $\lambda_{1}=0,
\lambda_{2},\cdots ,\lambda_{n}$, then from {Theorem~\ref{thm:eigenPST}}
we know that for $\ell =1,\cdots ,n$, $\lambda_{\ell }\equiv 1-h_{j
\ell }h_{k\ell } \mod {4}$. Therefore the eigenvalues $0, n-\lambda
_{2},\allowbreak \ldots , n-\lambda_{n}$ of $G^{c}$ satisfy $(n-\lambda_{\ell })
\equiv -(1-h_{j\ell }h_{k\ell })\equiv 1-h_{j\ell }h_{k\ell } \mod
{4}$, since $1-h_{j\ell }h_{k\ell }$ is either $0$ or $2 \,
\operatorname{mod} 4$ and $n$ must be a multiple of $4$ in order for a
Hadamard of order $n$ to exist.

Again from {Theorem~\ref{thm:eigenPST}}, we then know that $G^{c}$ has PST
from vertex $j$ to $k$ at time~$\pi /2$. Thus
$G\vee G=(G^{c}+G^{c})^{c}$ also has PST from vertex $j$ to $k$ at time
$\pi /2$.
\end{proof}
Note that we can also prove that $G^{c}$ exhibits PST at $t_{0}=
\pi /2$ by noticing that if $\lambda $ is a nonzero Laplacian eigenvalue
for $G$, then $n-\lambda $ is a Laplacian eigenvalue for~$G^{c}$ with
the same eigenvector. As $n-\lambda \equiv \lambda \mod {4}$, the
conclusion now follows from {Theorem~\ref{thm:eigenPST}}.\looseness=1

\goodbreak

We now introduce a modification of $G_{1} \ltimes G_{2}$ that, much like
$G_{1} \ltimes G_{2}$, can be used to construct new graphs with PST from
old ones. Suppose that $G_{1}$ and $G_{2}$ are two weighted graphs of
order $n$, with Laplacians $L_{1} = D_{1}-A_{1}$ and $L_{2} = D_{2}-A
_{2}$, respectively. Then we define the \emph{merge} of $G_{1}$ and
$G_{2}$ with respect to the weights $w_{1}$ and ${{w_{2}}}$ to be the graph
 $G_1 \tensor[_{w_1}]{\odot}{_{w_2}} G_2$ with Laplacian
\begin{eqnarray*}
\begin{bmatrix}
w_{1}L_{1} + {{w_{2}}}D_{2} & -{{w_{2}}}A_{2}
\\ \noalign{\vspace{2pt}}
-{{w_{2}}}A_{2} & w_{1}L_{1} + {{w_{2}}}D_{2}
\end{bmatrix}
.
\end{eqnarray*}
In the case that $w_{1} = {{w_{2}}} = 1$, we denote the merge simply by
$G_{1} \odot G_{2}$, and it recovers $G_{1} \ltimes G_{2}$.

\begin{figure}[htb]
	\captionsetup{justification=raggedright,singlelinecheck=false}
	\begin{center}
		\begin{tikzpicture}[x=2.4cm, y=2.4cm, label distance=0cm]
		\tikzset{
			bracem/.style={
				decoration={brace, mirror},
				decorate
			},
			brace/.style={
				decoration={brace},
				decorate
			}
		}
		
		\vertex[fill] (w0) at (0,0) [label=210:$1$]{};
		\vertex[fill] (w1) at (0,1) [label=120:$2$]{};
		\vertex[fill] (w2) at (1,1) [label=60:$3$]{};
		\vertex[fill] (w3) at (1,0) [label=300:$4$]{};
		
		\vertex[fill] (v0) at (0.5,-1.8) [label=210:$1$]{};
		\vertex[fill] (v1) at (0.5,-0.8) [label=120:$2$]{};
		\vertex[fill] (v2) at (1.5,-0.8) [label=60:$3$]{};
		\vertex[fill] (v3) at (1.5,-1.8) [label=300:$4$]{};
		
		\path
		(v0) edge (v1)
		(v0) edge (v2)
		(v3) edge (v1)
		(v3) edge (v2)
		
		(w0) edge (w1)
		(w1) edge (w2)
		(w2) edge (w3)
		(w3) edge (w0)
		;
		
		\vertex[fill] (x0) at (3.5,0) [label=210:$1$]{};
		\vertex[fill] (x1) at (3.5,1) [label=120:$2$]{};
		\vertex[fill] (x2) at (4.5,1) [label=60:$3$]{};
		\vertex[fill] (x3) at (4.5,0) [label=355:$4$]{};
		
		\vertex[fill] (y0) at (4,-1.8) [label=210:$5$]{};
		\vertex[fill] (y1) at (4,-0.8) [label=175:$6$]{};
		\vertex[fill] (y2) at (5,-0.8) [label=60:$7$]{};
		\vertex[fill] (y3) at (5,-1.8) [label=300:$8$]{};
		
		\path
		(x0) edge (x1)
		(x1) edge (x2)
		(x2) edge (x3)
		(x3) edge (x0)
		
		(y0) edge (y1)
		(y1) edge (y2)
		(y2) edge (y3)
		(y3) edge (y0)
		
		(x0) edge[dashed,color=gray] (y1)
		(x0) edge[dashed,color=gray] (y2)
		(x1) edge[dashed,color=gray] (y0)
		(x1) edge[dashed,color=gray] (y3)
		(x2) edge[dashed,color=gray] (y0)
		(x2) edge[dashed,color=gray] (y3)
		(x3) edge[dashed,color=gray] (y1)
		(x3) edge[dashed,color=gray] (y2)
		;
		
		\draw[->] (2,-0.4) -- (3,-0.4);
		\draw [bracem] (1.9,-2) -- (1.9,1.2);
		\draw [brace] (3.1,-2) -- (3.1,1.2);
		\end{tikzpicture}
	\end{center}
	\caption{A depiction of two Hadamard diagonalizable graphs (left) and their merge (right). The new graph has two copies of the original vertex set, and there is now an edge $(j,k)$ and $(n+j,n+k)$ if and only if $G_1$ (top left) had edge $(j,k)$, and there is an edge $(j,n+k)$ if and only if $G_2$ (bottom left) had edge $(j,k)$.}\label{fig:new_graph_operation}
\end{figure}
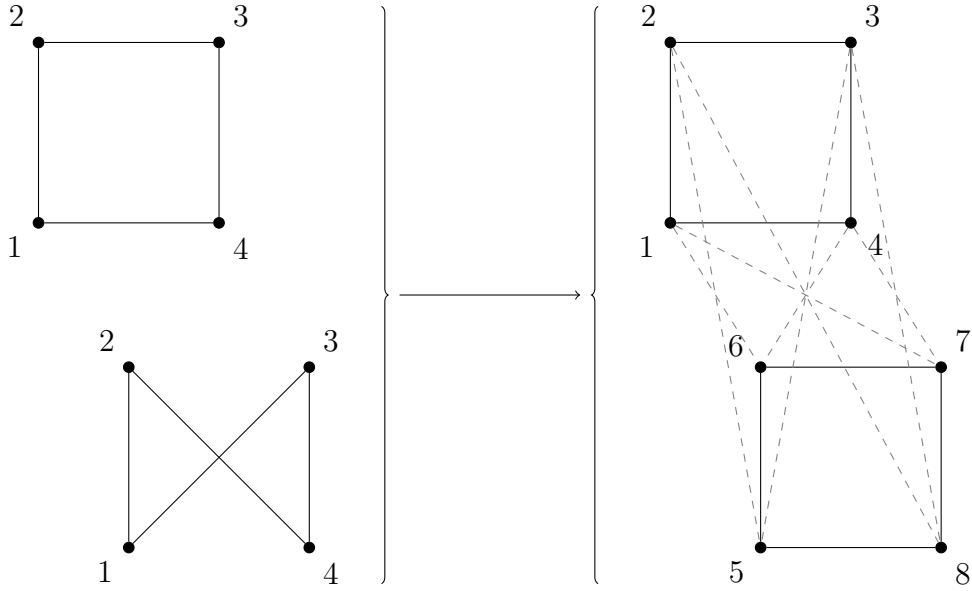

Observe that if $G_{1}$ and $G_{2}$ are both diagonalizable by the same
Hadamard matrix~$H$, then $G_{1} \tensor[_{w_1}]{\odot }{_{{{w_{2}}}}} G
_{2}$ is also Hadamard diagonalizable, by the matrix $
\begin{bmatrix}
H & H
\\
H & -H
\end{bmatrix}
$; this observation is what motivates our definition of the merge. While
this operation is a bit less intuitive than the other ones we saw, it
does have an interpretation in terms of the vertices and edges of the
original graphs. Specifically, if $G_{1}$ and $G_{2}$ each have vertices
labelled $\{1,\ldots ,n\}$, then $G_{1}
\tensor[_{w_{1}}]{\odot }{_{{{w_{2}}}}} G_{2}$ has twice as many vertices,
which we label $\{1,\ldots ,2n\}$. Furthermore, if $G_{1}$ has edge
$(j,k)$ with weight $w_{jk}$ then $G_{1}
\tensor[_{w_{1}}]{\odot }{_{{w_{2}}}} G_{2}$ has edges $(j,k)$ and
$(n+j,n+k)$, each with weight $w_{1}w_{jk}$. Similarly, if $G_{2}$ has
edge $(j,k)$ with weight $w_{jk}$ then $G_1\tensor[_{w_1}]{\odot}{_{w_2}} G_2$ has edge $(j,n+k)$ and
$(k, n+j)$ with weight ${{w_{2}}}w_{jk}$. See
{Fig.~\ref{fig:new_graph_operation}} for an example in the unweighted
case---the Laplacian matrices corresponding to $G_{1}$, $G_{2}$, and
$G_{1} \odot G_{2}$ in the example are, respectively,
\begin{eqnarray*}
L_{1}
& = &
\begin{bmatrix}
2 & -1 & 0 & -1
\\
-1 & 2 & -1 & 0
\\
0 & -1 & 2 & -1
\\
-1 & 0 & -1 & 2
\end{bmatrix}
, \quad L_{2} =
\begin{bmatrix}
2 & -1 & -1 & 0
\\
-1 & 2 & 0 & -1
\\
-1 & 0 & 2 & -1
\\
0 & -1 & -1 & 2
\end{bmatrix}
, \text{ and}
\\
L_{3}
& = &
\begin{bmatrix}
4 & -1 & 0 & -1 & 0 & -1 & -1 & 0
\\ \noalign{\vspace*{2pt}}
-1 & 4 & -1 & 0 & -1 & 0 & 0 & -1
\\ \noalign{\vspace*{2pt}}
0 & -1 & 4 & -1 & -1 & 0 & 0 & -1
\\ \noalign{\vspace*{2pt}}
-1 & 0 & -1 & 4 & 0 & -1 & -1 & 0
\\ \noalign{\vspace*{2pt}}
0 & -1 & -1 & 0 & 4 & -1 & 0 & -1
\\ \noalign{\vspace*{2pt}}
-1 & 0 & 0 & -1 & -1 & 4 & -1 & 0
\\ \noalign{\vspace*{2pt}}
-1 & 0 & 0 & -1 & 0 & -1 & 4 & -1
\\ \noalign{\vspace*{2pt}}
0 & -1 & -1 & 0 & -1 & 0 & -1 & 4
\end{bmatrix}
.
\end{eqnarray*}

We now describe an exact characterization of when the merge of two
integer-weighted graphs which are diagonalizable by the same Hadamard
matrix has PST at time $t_{0} = \pi /2$. This gives us a wide variety
of new graphs with PST; in particular, the merge operation produces
perfect state transfer graphs in a variety of scenarios. We note that
the result below can be proven by using techniques developed in
\cite{doublecover} for the adjacency matrix. However for completeness,
we give a separate proof that relies on Theorem~\ref{thm:eigenPST}.\looseness=1

\begin{theorem}
\label{thm:weighted_graph_merge}
Suppose $G_{1}$ and $G_{2}$ are integer-weighted graphs on $n$ vertices,
both of which are diagonalizable by the same Hadamard matrix $H$. Fix
$w_{1}, {{w_{2}}}\in \mathbb{Z}$ and let $L_{1}=d_{1}I-\nobreak A_{1}$, $L_{2}=d_{2}I-A_{2}$ be the Laplacian matrices for $G_{1}, G_{2}$, respectively. Then
$G_1\tensor[_{w_1}]{\odot}{_{w_2}} G_2$ has PST from vertex
$p$ to $q$, where $p<q$, at time $t_{0}=\pi /2$ if and only if one of
the following 8 conditions holds:
\begin{enumerate}
\item[1.] $p,q\in \{1,\ldots ,n\}$ and
\begin{enumerate}[(c)]
\item[(a)] $w_{1}$ is odd, ${{w_{2}}}$ is even, and $G_{1}$ has PST from $p$ to $q$ at
$t_{0}=\pi /2$, or
\item[(b)] $w_{1}$ and $d_{2}$ are even, ${{w_{2}}}$ is odd, and $G_{2}$ has PST from
$p$ to $q$ at $t_{0}=\pi /2$, or
\item[(c)] $w_{1}$ and ${{w_{2}}}$ are odd, $d_{2}$ is even, and the weighted graph
with Laplacian $L_{1}+L_{2}$ has PST from $p$ to $q$ at $t_{0}=\pi /2$;
\end{enumerate}

\item[2.] $p,q\in \{n+1,\ldots ,2n\}$ and
\begin{enumerate}[(c)]
\item[(a)] $w_{1}$ is odd, ${{w_{2}}}$ is even, and $G_{1}$ has PST from $p-n$ to
$q-n$ at $t_{0}=\pi /2$, or
\item[(b)] $w_{1}$ and $d_{2}$ are even, ${{w_{2}}}$ is odd, and $G_{2}$ has PST from
$p-n$ to $q-n$ at $t_{0}=\pi /2$,~or
\item[(c)] $w_{1}$ and ${{w_{2}}}$ are odd, $d_{2}$ is even, and the weighted graph
with Laplacian $L_{1}+L_{2}$ has PST from $p-n$ to $q-n$ at
$t_{0}=\pi /2$;
\end{enumerate}

\item[3.] $p\in \{1,\ldots ,n\}$, $q\in \{n+1,\ldots ,2n\}$ and
\begin{enumerate}[(b)]
\item[(a)] $w_{1}$ is even, ${{w_{2}}}$ and $d_{2}$ are odd, and $G_{2}$ has PST from
$p$ to $q-n$ at $t_{0}=\pi /2$, or
\item[(b)] $w_{1}$, ${{w_{2}}}$, and $d_{2}$ are all odd, and the weighted graph with
Laplacian matrix $L_{1}+L_{2}$ has PST from $p$ to $q-n$ at
$t_{0}=\pi /2$.
\end{enumerate}
\end{enumerate}
\end{theorem}
\begin{proof}
Without loss of generality we can assume that $H$ is a normalized
Hadamard matrix. Denote the diagonal matrices of eigenvalues for
$L_{1}, L_{2}$ by $\Lambda_{1}, \Lambda_{2}$, respectively, so that
$L_{j}=\frac{1}{n}H\Lambda_{j} H^{T}$, $j=1, 2$. Then the Laplacian of
$G_{1}\tensor[_{w_{1}}]{\odot }{_{{{w_{2}}}}} G_{2}$ is $L_{3}=
\begin{bmatrix}
w_{1}L_{1}+{{w_{2}}}d_{2}I &-{{w_{2}}}A_{2}
\\
-{{w_{2}}}A_{2} & w_{1}L_{1}+{{w_{2}}}d_{2}I
\end{bmatrix}
$. Further,
\begin{equation*}
L_{3}=\frac{1}{2n}
\begin{bmatrix}
H & H
\\
H &-H
\end{bmatrix}
\begin{bmatrix}
w_{1}\Lambda_{1}+{{w_{2}}}\Lambda_{2} &0
\\
0 & w_{1}\Lambda_{1}-{{w_{2}}}\Lambda_{2}+2{{w_{2}}}d_{2}I
\end{bmatrix}
\begin{bmatrix}
H & H
\\
H &-H
\end{bmatrix}
^{T}.
\end{equation*}
Denote the eigenvalues of $L_{1}, L_{2}$ by $\lambda_{\ell }^{(1)},
\lambda_{\ell }^{(2)}$, $\ell =1,\cdots , n$, respectively.
\begin{enumerate}
\item[1.] Suppose that $p,q\in \{1,\cdots ,n\}$ and that the graph with Laplacian $L_{3}$ has PST from $p$ to
$q$. Then for each $\ell =1,\cdots ,n$, $w_{1}\lambda_{\ell }^{(1)}+w
_{2}\lambda_{\ell }^{(2)}\equiv (1-h_{p\ell }h_{q\ell })\mod {4}$ and
$w_{1}\lambda_{\ell }^{(1)}-{{w_{2}}}\lambda_{\ell }^{(2)}+2{{w_{2}}}d_{2}
\equiv (1-h_{p\ell }h_{q\ell })\mod {4}$. In particular, $2{{w_{2}}}d_{2}
\equiv 0\mod {4}$, i.e., ${{w_{2}}}d_{2}$ is even. Note that if
$w_{1}$ and ${{w_{2}}}$ are both even, then $h_{p\ell }h_{q\ell }=1$ for
$\ell =1,\cdots ,n$, which is impossible.

If $w_{1}$ is odd and ${{w_{2}}}$ is even, then $w_{1}\lambda_{\ell }^{(1)}+w
_{2}\lambda_{\ell }^{(2)}\equiv \lambda_{\ell }^{(1)}\mod {4}$, so that
$\lambda_{\ell }^{(1)}\equiv (1-h_{p\ell }h_{q\ell })\mod {4}$,
$\ell =1,\cdots ,n$. Hence $G_{1}$ has PST from $p$ to $q$. Similarly,
if $w_{1}$ is even and ${{w_{2}}}$ is odd, then necessarily $d_{2}$ is even,
and as above $G_{2}$ has PST from $p$ to~$q$.

If $w_{1}$ and ${{w_{2}}}$ are both odd, then necessarily $d_{2}$ is even.
Also $w_{1}\lambda_{\ell }^{(1)}+{{w_{2}}}\lambda_{\ell }^{(2)}\equiv
\lambda_{\ell }^{(1)}+\lambda_{\ell }^{(2)}\equiv (1-h_{p\ell }h_{q
\ell })\mod {4}$, $\ell =1,\cdots ,n$. We deduce that the graph with Laplacian $L_{1}+L_{2}$ has
PST from $p$ to $q$.
\item[2.] If $p,q\in \{n+1,\cdots ,2n\}$ and the graph with Laplacian  $L_{3}$ has PST from $p$ to $q$, the
conclusions (a), (b), and (c) follow analogously to Case~1 above.
\item[3.] Suppose that $p\in \{1,\cdots , n\}$, $q\in \{n+1, \cdots , 2n\}$ and that the graph with Laplacian  
$L_{3}$ has PST from $p$ to $q$. Set $\hat{q}=q-n$. Then for each
$\ell =1,\cdots , n$, we have
\begin{eqnarray}\label{eq:cond1}
w_1\lambda_\ell^{(1)}+w_2\lambda_\ell^{(2)}\equiv(1-h_{p\ell}h_{\hat{q}\ell})\mod{4} \text{, and}
\end{eqnarray}
\vspace{-1cm}
\begin{eqnarray}\label{eq:cond2}
w_1\lambda_\ell^{(1)}-w_2\lambda_\ell^{(2)}+2w_2d_2\equiv(1+h_{p\ell}h_{\hat{q}\ell})\mod{4}.
\end{eqnarray}
Summing equations {(\ref{eq:cond1})} and {(\ref{eq:cond2})}, we find that
$2w_{1}\lambda_{\ell }^{(1)}+2{{w_{2}}}d_{2}\equiv 2\mod {4}$ and hence
$2{{w_{2}}}d_{2}\equiv 2\mod {4}$ since all the eigenvalues of $L_{1}$ are
even integers, and therefore ${{w_{2}}}d_{2}$ must be odd, i.e.,
${{w_{2}}}$ is odd and $d_{2}$ is odd. We have the following two cases.

If $w_{1}$ is even, then {(\ref{eq:cond1})} simplifies to $\lambda_{
\ell }^{(2)}\equiv (1-h_{p\ell }h_{\hat{q}\ell })\mod {4}$,
$\ell =1,\cdots ,n$, so for even $w_{1}$, and odd ${{w_{2}}}$ and
$d_{2}$, $G_{2}$ has PST from $p$ to $\hat{q}$.

If $w_{1}$ is odd, then {(\ref{eq:cond1})} simplifies to $\lambda_{
\ell }^{(1)}+\lambda_{\ell }^{(2)}\equiv (1-h_{p\ell }h_{\hat{q}
\ell })\mod {4}$, $\ell =1,\cdots ,n$, which shows that the
integer-weighted graph with Laplacian $L_{1}+L_{2}$ has PST from~$p$ to
$\hat{q}$.
\end{enumerate}
The converses are straightforward.
\end{proof}
Note that when both $w_{1}$ and ${{w_{2}}}$ are even, the graph
$G_{1}\tensor[_{w_{1}}]{\odot }{_{{{w_{2}}}}} G_{2}$ does not have PST at
time $\pi /2$. However, it might have PST at some other time. To see
this, we decompose the two integer weights $w_{j}$ as $w_{j}=2^{r_{j}}.b
_{j}$ (for $j=1,2$), where $b_{j}$ are odd integers. Let $r=\min (r
_{1}, r_{2})$. Then the PST property of the graph with Laplacian
$\frac{1}{2^{r}}L_{3}$ at time $\pi /2$ can be determined according to
{Theorem~\ref{thm:weighted_graph_merge}}. In the case that PST occurs, the
graph $G_{1}\tensor[_{w_{1}}]{\odot }{_{{{w_{2}}}}} G_{2}$ would then have
PST at time $\pi /2^{r+1}$. Also note that {Theorem~\ref{thm:weighted_graph_merge}} is true for any graphs whose Laplacian
eigenvalues are all even integers (including non integer-weighted
graphs).

\begin{remark}\label{suffiPST}
Assume that $G_{1}$ and $G_{2}$ are two graphs on $2^{m}$
vertices for $m\geq 2$ and that they are diagonalizable by the same
Hadamard matrix. Suppose that $G_{1}$ has PST from vertex $p$ to vertex
$q$, and $G_{2}$ has all its eigenvalues being multiples of 4 and that
its degree $d_{2}$ is odd (for example, a disjoint union of
$2^{m-r}$ copies of $K_{2^{r}}$ for $2\leq r\leq m$). Then $G_{1}
\odot G_{2}$ has PST from $p$ to $q+2^{m}$ according to Case~3(b) in
{Theorem~\ref{thm:weighted_graph_merge}}. Similarly, $G_{2} \odot G_{1}$
has PST from vertex $p$ to $q$ if $d_{1}$ is even (Case~1(c)), and it
has PST from vertex $p$ to $q+2^{m}$ if $d_{1}$ is odd (Case~3(b)).
\end{remark}

The requirement that both graphs are diagonalizable by the same Hadamard
matrix is necessary for {Theorem~\ref{thm:weighted_graph_merge}} to hold.
As a concrete example, let $G_{1}$ be equal to $K_{8}$ with a
$K_{3}$ removed, $G_{2}$ be equal to the 3-cube, $w_{1}=2$ and
${{w_{2}}}=1$ (and $d_{2}=3$). Then $G_{1}
\tensor[_{w_{1}}]{\odot }{_{{{w_{2}}}}} G_{2}$ is equal to
\begin{eqnarray*}
\begin{bmatrix}
13 & 0 & 0 & -2 & -2 & -2 & -2 & -2 & 0 &-1 & -1 & 0 & -1 & 0 & 0 & 0
\\
0 & 13 & 0 & -2 & -2 & -2 & -2 & -2 & -1 & 0 & 0 & -1 & 0 & -1 & 0 & 0
\\
0 & 0 & 13 & -2 & -2 & -2 & -2 & -2 & -1 & 0 & 0 & -1 & 0 & 0 & -1 & 0
\\
-2 & -2 & -2& 17 & -2 & -2 & -2 & -2 & 0 & -1 & -1 & 0 & 0 & 0 & 0 & -1
\\
-2 & -2 & -2 & -2 & 17 & -2 & -2 & -2 & -1 & 0 & 0 &0 & 0 & -1 & -1 &
0
\\
-2 & -2 & -2 & -2 & -2 & 17 & -2 & -2 & 0 & -1 & 0 & 0 & -1 & 0 & 0 &
-1
\\
-2 & -2 & -2 & -2 & -2 & -2 & 17 & -2 & 0 & 0 & -1 & 0 & -1 & 0 & 0 &
-1
\\
-2 & -2 & -2 & -2 & -2 & -2 & -2 & 17 & 0 & 0 & 0 & -1 & 0 & -1 & -1 &
0
\\
0 & -1 & -1 & 0 & -1 & 0 & 0 & 0 & 13 & 0 & 0 & -2 & -2 & -2 & -2 & -2
\\
-1 & 0 & 0 & -1 & 0 & -1 & 0 & 0 & 0 & 13 & 0 & -2 & -2 & -2 & -2 & -2
\\
-1 & 0 & 0 & -1 & 0 & 0 & -1 & 0 & 0 & 0 & 13 & -2 & -2 & -2 & -2 & -2
\\
0 & -1 & -1 & 0 & 0 & 0 & 0 & -1 & -2 & -2 & -2 & 17 & -2 & -2 & -2 &
-2
\\
-1 & 0 & 0 & 0 & 0 & -1 & -1 & 0 & -2 & -2 & -2 & -2 & 17 & -2 & -2 &
-2
\\
0 & -1 & 0 & 0 & -1 & 0 & 0 & -1 & -2 & -2 & -2 & -2 & -2 & 17 & -2 &
-2
\\
0 & 0 & -1 & 0 & -1 & 0 & 0 & -1 & -2 & -2 & -2 & -2 & -2 & -2 & 17 &
-2
\\
0 & 0 & 0 & -1 & 0 & -1 & -1 & 0 & -2 & -2 & -2 & -2 & -2 & -2 & -2 &
17
\end{bmatrix}
.
\end{eqnarray*}
There is no PST at time $\pi /2$, though the parameters are set up so
that they satisfy~3(a) of {Theorem~\ref{thm:weighted_graph_merge}} (but
not the hypothesis of both Laplacians being diagonalized by the same
Hadamard). Thus, unlike a similar result
\cite[Theorem~5.2]{doublecover} for the ``$\ltimes $'' operation (which
uses the adjacency matrices), graphs whose Laplacian matrices aren't
diagonalizable by the same Hadamard matrix do not necessarily satisfy
the conclusion of the theorem. This may be due to the difference between
Laplacian dynamics and adjacency dynamics.\looseness=1

The following corollary to {Theorem~\ref{thm:weighted_graph_merge}}
provides an instance where the statement of the theorem simplifies
considerably, and generalizes the known fact that the unweighted
hypercube graph has PST.

\begin{corollary}
\label{cor:k_cube_weight}
Suppose $w_{1},{{w_{2}}},\ldots ,w_{n}$ are nonzero integers, exactly
$d$ of which are odd, and consider the weighted hypercube $C_{n} := (w
_{1} K_{2}) \square ({{w_{2}}} K_{2}) \square \cdots \square (w_{n} K_{2})$.
For each vertex $u$ of $C_{n}$, there is a vertex $v$ at distance
$d$ from $u$ such that there is perfect state transfer in $C_{n}$ from
$u$ to $v$ at time $t_{0} = \pi /2$.
\end{corollary}
\begin{proof}
We prove the result by induction on $n$. For the base case, we simply
note that it is straightforward to verify that the weighted $1$-cube
$w_{1}K_{2}$ has perfect state transfer at time $t = \pi /2$ if and only
if $w_{1}$ is an odd integer.

For the inductive hypothesis, we use
{Theorem~\ref{thm:weighted_graph_merge}} with $G_{1} = C_{n}$ (which we
will assume has perfect state transfer at time $t = \pi /2$ from vertex
$j$ to $k$, which are a distance of $d$ apart) and $G_{2}$ is the graph
on the same number of vertices where every vertex has a self-loop (of
weight $1$) and no other edges (note that this graph has perfect state
transfer between any vertex and itself at any time). Then it is
straightforward to verify that the graph $G_{1}
\tensor[_{1}]{\odot }{_{w_{n+1}}} G_{2}$ is exactly the weighted
$(n+1)$-cube:
\[
G_{1}\tensor[_{1}]{\odot }{_{w_{n+1}}} G_{2} = (w_{n+1} K_{2}) \square
C_{n} = (w_{n+1} K_{2}) \square (w_{1} K_{2}) \square ({{w_{2}}} K_{2})
\square \cdots \square (w_{n} K_{2}).
\]
So condition~1(a) of {Theorem~\ref{thm:weighted_graph_merge}} tells us
that if $w_{n+1}$ is even then $G_{1}
\tensor[_{1}]{\odot }{_{w_{n+1}}} G_{2}$ has perfect state transfer at
time $t_{0} = \pi /2$ from vertex $j$ to $k$ (which still have a
distance of $d$ from each other). On the other hand, if $w_{n+1}$ is odd
then condition~3(b) of {Theorem~\ref{thm:weighted_graph_merge}} says that
$G_{1}\tensor[_{1}]{\odot }{_{w_{n+1}}} G_{2}$ has perfect state
transfer at time $t = \pi /2$ from vertex $j$ to $k+2^{n}$ (which have
a distance of $d+1$ from each other). By noting that the particular
{labelling} of the weights is irrelevant (i.e., permute the indices of the
weights so that $G_{1}\tensor[_{1}]{\odot }{_{w_{n+1}}}G_{2} = (w_{1}
K_{2}) \square ({{w_{2}}} K_{2}) \square \cdots \square (w_{n+1} K_{2}) =
C_{n+1}$), this completes the inductive step and the proof.
\end{proof}

\begin{example}
From Lemma~9 and Proposition~10 of \cite{ShaunSteve}, one
can conclude that there is no unweighted graph of order 12 that is
Hadamard diagonalizable and exhibits PST. However, it is easy to
construct weighted graphs of this type. Let $G_{1}$ be the graph whose
Laplacian is
\begin{eqnarray*}
L_{1}=\frac{1}{3}
\begin{bmatrix}
18 & 0 & -1 & -1 & -1 & -3 & -3 & -3 & -1 & -3 & -1 & -1
\\ \noalign{\vspace*{2pt}}
0 &18 & -1 & -1 & -1 & -3 & -3 & -3 & -1 & -3 & -1 & -1
\\ \noalign{\vspace*{2pt}}
-1 &-1 & 18& -2 & -2 & 0 & -2 & 0 &-2 & -2 & -4 & -2
\\ \noalign{\vspace*{2pt}}
-1&-1& -2 & 18 & -4 & 0 & 0 & -2 & -2 & -2 & -2 & -2
\\ \noalign{\vspace*{2pt}}
-1&-1& -2 & -4 & 18 & -2 & -2 & 0 & -2 & 0 & -2 & -2
\\ \noalign{\vspace*{2pt}}
-3&-3 & 0 & 0 & -2 & 18 & -2 & -2 & 0 & -2 & -2 & -2
\\ \noalign{\vspace*{2pt}}
-3& -3& -2 & 0 & -2 & -2 & 18 & -2 & -2 & -2 & 0& 0
\\ \noalign{\vspace*{2pt}}
-3&-3 & 0 & -2 & 0 & -2 & -2 & 18 & -2 & -2 & -2& 0
\\ \noalign{\vspace*{2pt}}
-1& -1 & -2 & -2 & -2 & 0 & -2 & -2 & 18 & 0 & -2 & -4
\\ \noalign{\vspace*{2pt}}
-3& -3 & -2 & -2 & 0 & -2 & -2 & -2 & 0 & 18 & 0 & -2
\\ \noalign{\vspace*{2pt}}
-1& -1& -4 & -2 & -2 & -2 & 0 & -2 & -2 & 0 & 18 & -2
\\ \noalign{\vspace*{2pt}}
-1& -1 & -2& -2 & -2 & -2& 0& 0&-4 & -2 & -2 & 18
\end{bmatrix}
\end{eqnarray*}
Then one can easily verify that $L_{1}$ is Hadamard diagonalizable by
the order 12 Hadamard
\begin{eqnarray*}
H=
\begin{bmatrix}
1 & 1 & 1 & 1 & 1 & 1 & 1 & 1 & 1 & 1 & 1 & 1
\\ \noalign{\vspace*{1pt}}
1 & -1 & 1 & -1 & 1 & 1 & 1 & -1 & -1 & -1 & 1 & -1
\\ \noalign{\vspace*{1pt}}
1 & -1 & -1 & 1 & -1 & 1 & 1 & 1 & -1 & -1 & -1 & 1
\\ \noalign{\vspace*{1pt}}
1 & 1 & -1 & -1 & 1 & -1 & 1 & 1& 1 & -1 & -1 & -1
\\ \noalign{\vspace*{1pt}}
1 & -1 & 1 & -1 & -1 & 1 & -1 & 1 & 1 & 1 & -1 & -1
\\ \noalign{\vspace*{1pt}}
1 & -1 & -1 & 1 & -1 & -1 & 1 & -1 & 1& 1 & 1 & -1
\\ \noalign{\vspace*{1pt}}
1 & -1 & -1 & -1 & 1 & -1 & -1 & 1 & -1 & 1 & 1 & 1
\\ \noalign{\vspace*{1pt}}
1 & 1 & -1 & -1 & -1 & 1 & -1 & -1 & 1 & -1 & 1 & 1
\\ \noalign{\vspace*{1pt}}
1 & 1 & 1 & -1 & -1 & -1 & 1 & -1& -1 & 1 & -1 & 1
\\ \noalign{\vspace*{1pt}}
1 & 1 & 1 & 1 & -1 & -1 & -1 & 1 & -1 & -1 & 1 & -1
\\ \noalign{\vspace*{1pt}}
1 & -1 & 1 & 1 & 1 & -1 & -1 & -1& 1 & -1 & -1 & 1
\\ \noalign{\vspace*{1pt}}
1 & 1 & -1& 1 & 1 & 1 & -1 & -1 & -1 & 1 & -1 & -1
\end{bmatrix}
\end{eqnarray*}
and that the $(1,2)$ entry of $e^{i(\pi /2)L_{1}}$ is 1, thus showing
that $L_{1}$ exhibits PST between vertices 1 and 2 at time $t_{0}=
\pi /2$. Let $G_{2}=K_{12}$, which we note is Hadamard diagonalizable
by $H$ but does not exhibit PST, and let $w_{1}=5$ and ${{w_{2}}}=2$. Direct
computation shows that all the eigenvalues of $L_{1}$ are even integers.
Hence {Theorem~\ref{thm:weighted_graph_merge}} still applies here, and
Case~1(a) of the theorem tells us that $G_{1}
\tensor[_{5}]{\odot }{_2} G_{2}$ has PST from vertex 1 to vertex~2 at
time $t_{0}=\pi /2$. One can indeed verify that $L_{3}=
\begin{bmatrix}
5L_{1} + 2D_{2} & -2A_{2}
\\
-2A_{2} & 5L_{1} + 2D_{2}
\end{bmatrix}
$, where $D_{2}=11I$ and $A_{2}=J-I$ (where $J$ is the all-ones matrix),
is Hadamard diagonalizable by $
\begin{bmatrix}
H& H
\\
H&-H
\end{bmatrix}
$\vspace{2pt} with eigenvalues (in the order determined by that diagonalization)
equal to $0$, $ 54$, $ 64$, $ 54$, $ 64$, $64$, $64$, $54$, $54$,
$54$, $44$, $54$, $44$, $50$, $60$, $50$, $ 60$, $ 60$, $60$, $ 50$,
$50$, $50$, $40$, and~$50$. Furthermore, by checking the $(1,2)$ entry
of $e^{i(\pi /2)L_{3}}$ we see that this graph exhibits PST between
vertices 1 and 2 at time $t_{0}=\pi /2$.
\end{example}

\begin{remark}\label{cubelikeconstruction}
For each $k\geq 3$ and each $d$ with $k+1\leq d\leq 2^{k}-2$, we
can construct a graph that is $d$-regular, unweighted, connected, and
non-bipartite on $2^{k}$ vertices, that is diagonalizable by the standard
Hadamard matrix and has PST at time $t_{0}=\pi /2$. This can be done
with cubelike graphs by using {Theorem~\ref{cubelikePST}}. To ensure the
cubelike graph is connected, we just need to make sure the connection
set contains a basis of $\mathbb{Z}_{2}^{k}$ when considered as a vector space. Let us take the standard
ordered basis: $\mathbf{e}_{1},\cdots ,\mathbf{e}_{k}$. Assume they form
the set $T$. For $d=k+1$, take the connection set $C=T\cup \{
\mathbf{e}_{1}+\mathbf{e}_{2}\}$. Then the induced subgraph on vertices
$0, \mathbf{e}_{1}, \mathbf{e}_{2}, \mathbf{e}_{1}+\mathbf{e}_{2}$ is
$K_{4}$; hence the corresponding cubelike graph is not bipartite. Also
note that the sum of the elements in $C$ is not 0 for $k\geq 3$. For
$d>k+1$, we always keep $C$ as a subset of the connection set $S$
($|S|=d, 0 \notin S$). If the sum of all elements in $S$ is not 0, then
the cubelike graph $G(S)$ is a desired graph. On the other hand, if the
sum of all elements in $S$ is 0, then we delete some element
$c_{0}$ from the set $S\setminus C$. Denote the set $S\setminus \{c
_{0}\}$ by $S_{0}$ and we know the sum of all its elements is
$c_{0}\neq 0$ (in $\mathbb{Z}_{2}^{k}$, every element has itself as its
inverse). Finally, we pick any element $c_{1}\in \mathbb{Z}_{2}^{k}
\setminus (S \cup \{0\})$ (this set has cardinality $2^{k}-d-1>0$) and
form a new set $S_{1}=S_{0}\cup \{c_{1}\}$. Then $S_{1}$ has cardinality~$d$
and the sum of all its element is $c=c_{0}+c_{1}\neq 0$. Hence there
is PST from $u$ to $u+c$ at time $\pi /2$ in the connected (since
$S_{1}$ is a generating set of the group $\mathbb{Z}_{2}^{k}$)
nonbipartite cubelike graph $G(S_{1})$.
\end{remark}

This remark can be stated as follows. As a means of highlighting the
utility of the merge operation, we present an alternate proof that
constructs such graphs using the merge.

\begin{theorem}
Suppose that $k\in \mathbb{N}$ with $k\geq 3$. For each $d\in
\mathbb{N}$ with $k+1\leq d\leq 2^{k}-2$, there is a connected,
unweighted, non-bipartite graph that is
\begin{enumerate}[(3)]
\item[(1)] diagonalizable by the standard Hadamard matrix of order $2^{k}$,
\item[(2)] $d$-regular, and
\item[(3)] has PST between distinct vertices at time $t_{0} = \pi /2$.
\end{enumerate}
\end{theorem}

\begin{proof}
For every integer $k\geq 3$, it is easy to see that the complement of
a perfect matching (disjoint union of $2^{k-1}$ copies of $K_{2}$) is
a $(2^{k}-2)$-regular graph with the desired properties. So we just need
to prove the result for $k+1\leq d\leq 2^{k}-3$. We proceed by induction
on $k$. For $k=3$, it is straightforward to check that $(K_{2,2}
\square K_{2})^{c}, (K_{2,2}+K_{2,2})^{c}$ and $(K_{2}+K_{2}+K_{2}+K
_{2})^{c}$ are $4$-, $5$-, and $6$-regular graphs, respectively,  having the described
properties. Now suppose the result holds for some fixed $k\geq 3$; that
is, for each $d$ with $k+1\leq d\leq 2^{k}-2$, we have a $d$-regular
graph $G_{k,d}$ on $2^{k}$ vertices with the desired properties. To
construct desired graphs on $2^{k+1}$ vertices, we split into two cases
depending on the regularity $d$ of the graph that we are trying to
construct.
\begin{enumerate}[Case~2:]
\item[Case~1:] $(k+1)+1 \leq d \leq 2^{k}-1$. Let the Laplacian matrix
for $G_{k,d-1}$ be $L_{k,d-1}$. Consider the graph $ K_{2}\square G
_{k,d-1}$. This graph\vspace{1.4pt} has $2^{k+1}$ vertices and Laplacian $
\begin{bmatrix}
L_{k,d-1}+I & -I
\\
-I & L_{k,d-1}+I
\end{bmatrix}
$.\vspace{1.4pt} It is straightforward to see that this graph is $d$-regular,
connected, non-bipartite (since $G_{k,d-1}$ is) and satisfies~(1)
and~(3).\looseness=1
\item[Case~2:] $k+4\leq d\leq 2^{k+1}-3$. For each $2\leq r\leq k$, let
$G_{r}$ be the disjoint union of $2^{k-r}$ copies of $K_{2^{r}}$ and let
$L_{r}$ denote its Laplacian matrix. Note that each eigenvalue of
$L_{r}$ is congruent to 0 (mod 4); further, with a natural labelling of
the vertices, $L_{r}$ is diagonalizable by the standard Hadamard matrix
of order $2^{k}$.\looseness=-1

Fix $d^{\prime }$ with $k+1\leq d^{\prime }\leq 2^{k}-2$ and let
$A_{k,d^{\prime }}$ be the adjacency matrix of~$G_{k,d^{\prime }}$. Let
$G_{k,d^{\prime }}^{(r)}=G_{r}\odot G_{k,d^{\prime }}$ be the graph on
$2^{k+1}$ vertices\vspace{1.4pt} whose Laplacian matrix is $L_{k,d^{\prime }}^{(r)}=
\begin{bmatrix}
L_{r}+d^{\prime }I & -A_{k,d^{\prime }}
\\
-A_{k,d^{\prime }} & L_{r}+d^{\prime }I
\end{bmatrix}
$.\vspace{1.4pt} Then $G_{k,d^{\prime }}^{(r)}$ is not bipartite (it has $K_{4}$ as
an induced subgraph), and it is Hadamard diagonalizable by the standard
Hadamard matrix. By {Remark~\ref{suffiPST}}, it also has PST between a
pair of distinct vertices. It is regular of degree
$d^{\prime }+2^{r}-1$. Also in the notation of {Theorem~\ref{thm:weighted_graph_merge}}, $-\Lambda_{2}+2d^{\prime }I$ has
positive diagonal entries (since $G_{k,d^{\prime }}$ is not bipartite)
and so we deduce that the nullity of $L_{k,d^{\prime }}^{(r)}$ is 1 and
that $G_{k,d^{\prime }}^{(r)}$ is connected. Thus $G_{k,d^{\prime }}
^{(r)}$ is a graph on $2^{k+1}$ vertices, satisfying the desired
properties. We denote it as $G_{k+1,d^{\prime }+2^{r}-1}$.\looseness=1
\end{enumerate}

Thus we have produced the desired graphs whose degrees fall in the set
%
\begin{eqnarray}
\label{eq:interval_unions}[k+2, 2^{k}-1]\cup \bigcup_{r=2}^{k}[k+2^{r},2^{k}+2^{r}-3].
\end{eqnarray}

For $k=3$, this set covers the integers 5, 6, 7, 8, 9, 11, 12, 13. From
{Proposition~\ref{compljoin}} we know that if $G=(K_{2,2}+K_{2,2})^{c}$,
our 5-regular graph on 8 vertices satisfying the desired properties,
then the graph $G^{c}\vee G^{c}$ of order 16 is 10-regular and has the
desired properties. So the result is also true for graphs on
$2^{4}$ vertices.

For $k\geq 4$ we have $k+4\leq 2^{k-1}$. Then for any $r\leq k-1$, we
have $k+4+2^{r+1}\leq 2^{k-1}+2^{r+1}=2^{k-1}+2^{r}+2^{r}
\leq 2^{k-1}+2^{k-1}+2^{r}=2^{k}+2^{r}$, which in turn implies that
$k+2^{r+1}\leq 2^{k}+2^{r}-3$. It follows that the
set~{\eqref{eq:interval_unions}} contains all of the integers in
$[k+2,2^{k+1}-3]$.
\end{proof}

\begin{remark}
Note that for $2^{k}+1\leq d\leq 2^{k+1}-2$, we can also
construct a $d$-regular graph with the desired properties on
$2^{k+1}$ vertices using the join operation. From the induction
hypothesis, we have a graph $G_{k,d}$ with the desired properties for
$k+1\leq d\leq 2^{k}-2$. Now we use the result in
{Proposition~\ref{compljoin}}: if $G$ is a Hadamard diagonalizable graph
on $n\geq 4$ vertices and that $G$ has PST at time $\pi /2$, then its
complement also has PST at the same time, then we get a non-empty
$d$-regular Hadamard diagonalizable PST graph $G$ for each $d$ such that
$1\leq d\leq 2^{k}-2$. Then the graph $G\vee G$ has PST at time
$t_{0} = \pi /2$ and is diagonalizable by the standard Hadamard matrix,
whose regularity is $2^{k}+d$, ranging from $2^{k}+1$ to
$2^{k}+2^{k}-2$. Since $G$ is not empty, $G\vee G$ has cycles of
length~$3$, and therefore it is not bipartite.
\end{remark}

\section{PST for graphs with non-integer weights}\label{sec:noninteger}
We now consider some ways in which our results generalize to the case
of Hadamard diagonalizable graphs with non-integer edge weights. In the
case where all of the edge weights are rational, the idea is rather
straightforward.

\begin{proposition}
\label{prop:rational}
Suppose the graph $G_{1}$ with Laplacian $L_{1}$ is a rational-weighted
Hadamard diagonalizable graph, and let $\mathop{lcm}$ be the least
common multiple of the denominators of its edge weights, and
$\mathop{gcd}$ be the greatest common divisor of all the new integer
edge weights $\mathop{lcm}\cdot w(j,k)$. Then $G_{1}$ has PST at time
$t_{1} = \frac{\mathop{lcm}}{\mathop{gcd}}\cdot \pi /2$ if and only if
the integer-weighted Hadamard diagonalizable graph $G_{2}$ with
Laplacian $L_{2} = \frac{\mathop{lcm}}{\mathop{gcd}}L_{1}$ has PST at
time $t_{0} = \pi /2$ between the same pair of vertices.
\end{proposition}

\begin{proof}
The result follows simply from noticing that for each $j$ and $k$ we
have
%
\begin{eqnarray}
\label{prf:rational}|\mathbf{e}_{j}^{T} e^{it_{0}L_{2}} \mathbf{e}_{k} |^{2} = |
\mathbf{e}_{j}^{T} e^{it_{0}\frac{\mathop{lcm}}{\mathop{gcd}}L_{1}}
\mathbf{e}_{k} |^{2} = |\mathbf{e}_{j}^{T} e^{it_{1}L_{1}} \mathbf{e}
_{k} |^{2},
\end{eqnarray}
and $G_{1}$ has PST between vertex $j$ and vertex $k$ at time
$t_{1}$ if and only the rightmost quantity in~{\eqref{prf:rational}}
equals $1$, while $G_{2}$ has PST at time $t_{0}$ if and only if the
leftmost quantity in~{\eqref{prf:rational}} equals $1$.
\end{proof}

While we are not able to extend {Proposition~\ref{prop:rational}} to the
case of irrational weights directly---in general such a graph may not
exhibit PST at any time---it is true at least that the resulting graph
has pretty good state transfer when exactly one of the two weights in
$G_{1} \tensor[_{w_{1}}]{\odot }{_{{{w_{2}}}}} G_{2}$ is irrational. Before
giving the theorem, we recall the following result about approximating
an irrational real number with rational numbers.

\begin{theorem}[\cite{dioph3seq}]
Let $o$ denote the odd integers and $e$ denote the even integers. Then
for every real irrational number $w$, there are infinitely many
relatively prime numbers $u,v$ with $[u,v]$ in each of the three classes
$[o,e]$, $[e,o]$, and $[o,o]$, such that the inequality $|w-u/v|<1/v
^{2}$ holds.
\end{theorem}
For the graph $G_{1} \tensor[_{w_{1}}]{\odot }{_{{{w_{2}}}}} G_{2}$, we say
it has parameters $[w_{1}, {{w_{2}}}, d_{2}]$, where as in {Theorem~\ref{thm:weighted_graph_merge}}, $d_{2}$ denotes the degree of
$G_{2}$. In particular, if $w_{1},{{w_{2}}}$, and $d_{2}$ are all odd
integers, we say the graph $G_{1} \tensor[_{w_{1}}]{\odot }{_{{w_{2}}}}
G_{2}$ has type $[o,o,o]$. We will denote the set of irrational numbers
by $\overline{\mathbb{Q}}$.

\begin{theorem}
\label{pgst_irrat}
Assume that $G_{1}$ and $G_{2}$ are integer-weighted graphs on $n$
vertices, both of which are diagonalizable by the same Hadamard matrix
$H$. Let $d_{2}$ be the degree of $G_{2}$. Let $L_{1}$ and $L_{2}$
denote the Laplacian matrices of $G_{1}$ and $G_{2}$, respectively.
Suppose that one of $w_{1}, {{w_{2}}}$ is rational and the other is
irrational, and suppose that $p,q \in \{1, \ldots , n\}$. Then the
weighted graph $G_{1}  \tensor[_{w_{1}}]{\odot }{_{{{w_{2}}}}} G_{2}$ has
PGST as stated in the following cases.
\begin{enumerate}[3.]
\item[1.] Suppose that $G_{1}$ has PST from $p$ to $q$ at time $\pi /2$. Then
$G_{1}  \tensor[_{w_{1}}]{\odot }{_{{w_{2}}}} G_{2}$ has PGST from $p$ to
$q$ and from $p+n$ to $q+n$.
\item[2.] Suppose that $G_{2}$ has PST from $p$ to $q$ at time $\pi /2$. If
$d_{2}$ is even, then $G_{1}  \tensor[_{w_{1}}]{\odot }{_{{{w_{2}}}}} G
_{2}$ has PGST from $p$ to $q$ and from $p+n$ to $q+n$. If $d_{2}$ is
odd, then $G_{1} \tensor[_{w_{1}}]{\odot }{_{w_2}} G_{2}$ has PGST
from $p$ to $q+n$ and from $q$ to $p+n$.
\item[3.] Suppose that the graph with Laplacian $L_{1}+L_{2}$ has PST from $p$ to $q$ at time $\pi /2$. If
$d_{2}$ is even, then $G_{1}  \tensor[_{w_{1}}]{\odot }{_{w_2}} G
_{2}$ has PGST from $p$ to $q$ and from $p+n$ to $q+n$. If $d_{2}$ is
odd, then $G_{1}  \tensor[_{w_{1}}]{\odot }{_{w_2}} G_{2}$ has PGST
from $p$ to $q+n$ and from $q$ to $p+n$.
\end{enumerate}
\end{theorem}
Before proving this result, we note that it can alternatively be proved
via Kronecker's theorem using the techniques of \cite{BCGS}.
However, this would require proving that vertices $p$ and $q$ are
strongly cospectral, as well as some knowledge of eigenvalues and
eigenprojection matrices, so we instead give the following proof that
is somewhat more self-contained.\looseness=1
\begin{proof}
As in the proof of {Theorem~\ref{thm:weighted_graph_merge}}, without loss
of generality we assume that $H$ is a normalized Hadamard matrix. Assume
$w_{1}$ is rational and ${{w_{2}}}$ is irrational (case 1). We denote the
graph $G_{1}  \tensor[_{w_{1}}]{\odot }{_{w_2}} G_{2}$ as $G_{3}$, with
corresponding Laplacian $L_{3}$. It suffices to consider $w_{1}$ odd.
Indeed, if $w_{1}=\frac{a}b$, with $a$ and $b$ being relatively prime integers, assume  $a=2^{r}k$ where $r\in
\mathbb{N}$ and $k$ odd, then $e^{itL_{3}}=e^{it\frac{2^{r}}{b}(L_{3}b/2^{r})}$ so
that $L_{3}b/2^{r}$ is the Laplacian of the graph $G_{1}
\tensor[_{k}]{\odot }{_{{{w_{2}}}b/2^{r}}} G_{2}$ where $k$ is odd and
${{w_{2}}}b/2^{r}$ is irrational. Note that if $L_{3}b/2^{r}$ has PGST, then
so does $L_{3}$. Thus, for notational {simplicity}, we consider
$w_{1}$ odd.

We approach ${{w_{2}}}$ with fractions $u/v$ such that $|{{w_{2}}}-u/v|<1/v
^{2}$. For each such pair of $u,v$, we denote the graph $G_{1}
 \tensor[_{w_{1}}]{\odot }{_{u/v}} G_{2}$ as $G_{4}$, and the graph
$G_{1}
\tensor[_0]{\odot }{_{{{w_{2}}}-u/v}} G_{2}$ as $G_{5}$. In
particular, the Laplacian of $G_{3}$ is the sum of the Laplacian of
$G_{4}$ with the Laplacian of $G_{5}$. Denote the Laplacian matrices of
$G_{4}$ and $G_{5}$ as $L_{4}$ and $L_{5}$, respectively. Now consider
the integer-weighted graph $G'_{4}=G_{1}
\tensor[_{vw_{1}}]{\odot }{_u} G_{2}$, then its Laplacian is
$vL_{4}$ and has parameters $[vw_{1}, u, d_{2}]$.

There are now a number of cases to consider. If $[u,v]$ is of type
$[o,e]$ and $d_{2}$ is even, the graph $G'_{4}$ is of type $[e,o,e]$.
From {Theorem~\ref{thm:weighted_graph_merge}} we know, if $G_{2}$ has PST
from $p$ to $q$ at $\pi /2$, then $G'_{4}$ has PST at $\pi /2$ from
$p$ to $q$ and from $p+n$ to $q+n$ (Case~1(b), 2(b)). If $[u,v]$ is of
type $[o,e]$ and $d_{2}$ is odd, the graph $G'_{4}$ is of type
$[e,o,o]$. From {Theorem~\ref{thm:weighted_graph_merge}} we know that if
$G_{2}$ has PST at $\pi /2$ from $p$ to $q$ at $\pi /2$, then
$G'_{4}$ has PST at $\pi /2$ from $p$ to $q+n$ and from $q$ to $p+n$
(Case~3(a)).

If $[u,v]$ is of type $[e,o]$, then the graph $G'_{4}$ is of type
$[o,e,f]$, where $f$ denotes the parity of $d_{2}$. From {Theorem~\ref{thm:weighted_graph_merge}} we know that if $G_{1}$ has PST from
$p$ to $q$ at $\pi /2$, then $G'_{4}$ has PST at $\pi /2$ from $p$ to
$q$ and from $p+n$ to $q+n$ (Case~1(a), 2(a)).

If $[u,v]$ is of type $[o,o]$ and $d_{2}$ is even, the graph
$G'_{4}$ is of type $[o,o,e]$. From {Theorem~\ref{thm:weighted_graph_merge}} we know that if the graph with Laplacian
$L_{1}+L_{2}$ has PST from $p$ to $q$ at $\pi /2$, then $G'_{4}$ has PST
from $p$ to $q$ and from $p+n$ to $q+n$ (Case~1(c), 2(c)). If
$[u,v]$ is of type $[o,o]$ and $d_{2}$ is odd, the graph $G'_{4}$ is of
type $[o,o,o]$. From {Theorem~\ref{thm:weighted_graph_merge}} we know that
if the integer weighted graph with Laplacian $L_{1}+L_{2}$ has PST from
$p$ to $q$ at $\pi /2$, then $G'_{4}$ has PST from $p$ to $q+n$ and from
$q$ to $p+n$ (Case~3(b)).

Similarly we can get the results when $w_{1}$ is irrational and
${{w_{2}}}$ is rational (case 2). (We can assume ${{w_{2}}}$ is odd by way of
a similar argument to $w_{1}$ being odd in case 1.)

For all the above cases, $G_{4}$ has PST at time $t_{0}=v\pi /2$. Next,
we recall the following result from \cite[Theorem~4]{GKLPZ} (here
we take the absorbed constant factor $t_{0}$ out): Suppose PST occurs
for the graph with Laplacian matrix $L$ and assume that $\hat{L}=t
_{0}(L+L_{0})$ due to a small nonzero edge-weight perturbation
$L_{0}$. Then
%
\begin{eqnarray}
\label{exp-bound}
1 - |\mathbf{e}_{j}^{T} e^{it_{0}(L+L_{0})} \mathbf{e}_{k} |^{2}
&
\le & 2\|t_{0}L_{0}\| + \|t_{0}L_{0}\|^{2} - \|t_{0}L_{0}\|^{3}.
\end{eqnarray}
Now, $G_{4}$ is a graph with PST at time $t_{0}$, and $L_{3}=L_{4}+L
_{5}$. Then the fidelity of state transfer of $G_{3}$ between the
corresponding pair of vertices satisfies
\begin{eqnarray*}
|\mathbf{e}_{j}^{T} e^{it_{0}L_{3}} \mathbf{e}_{k} |^{2}
& \geq & 1-2
\|t_{0}L_{5}\| - \|t_{0}L_{5}\|^{2} + \|t_{0}L_{5}\|^{3}
\\
& \geq & 1-2cn\pi /(2v)-(cn\pi /(2v))^{2}+(cn\pi /(2v))^{3},
\end{eqnarray*}
where $c$ is the maximum edge weight in $G_{2}$. Since there are
infinitely such integers $v$, the expression on the right hand side in the above
inequality can be made as close to one as possible.
\end{proof}

It is known (see \cite{Kay11}) that if there is perfect state
transfer from vertex $j$ to vertex $k$ at time $t_{0}$, and perfect
state transfer from vertex $j$ to vertex $l$ at time $t_{1}$, then
necessarily $k=l$. The following example, which is a straightforward
consequence of {Theorem~\ref{pgst_irrat}}, shows that the situation with
respect to pretty good state transfer is markedly different. This is a
potentially important application to routing---the task of choosing
between several possible recipients of the state.

\begin{example}
Consider the unweighted graphs $G_{1}, G_{2}$ with the following
Laplacian matrices:
\begin{equation*}
L_{1}= \left[
\begin{array}{c@{\quad}c@{\quad}c@{\quad}c@{\quad}c@{\quad}c@{\quad}c@{\quad}c}
3 & -1 & -1 & 0 & -1 & 0 & 0 & 0
\\ \noalign{\vspace*{2pt}}
-1 & 3 & 0 & -1 & 0 & -1 & 0 & 0
\\ \noalign{\vspace*{2pt}}
-1 & 0 & 3 & -1 & 0 & 0 & -1& 0
\\ \noalign{\vspace*{2pt}}
0 & -1 & -1 & 3 & 0 & 0 & 0 & -1
\\ \noalign{\vspace*{2pt}}
-1 & 0 & 0 & 0 & 3 & -1 & -1 & 0
\\ \noalign{\vspace*{2pt}}
0 & -1 & 0 & 0 & -1 & 3 & 0 & -1
\\ \noalign{\vspace*{2pt}}
0 & 0 & -1 & 0 & -1 & 0 & 3 & -1
\\ \noalign{\vspace*{2pt}}
0 & 0 & 0 & -1 & 0 & -1 & -1 & 3
\end{array}
\right]
\end{equation*}
which has PST at time $\pi /2$ for the pairs $(1,8), (2,7), (3,6),
(4,5)$, and
\begin{equation*}
L_{2}= \left[
\begin{array}{c@{\quad}c@{\quad}c@{\quad}c@{\quad}c@{\quad}c@{\quad}c@{\quad}c}
3 & 0 & -1 & -1 & -1 & 0 & 0 & 0
\\ \noalign{\vspace*{2pt}}
0 & 3 & -1 & -1& 0 & -1 & 0 & 0
\\ \noalign{\vspace*{2pt}}
-1 & -1 &3 &0 &0 &0 &-1 &0
\\ \noalign{\vspace*{2pt}}
-1 &-1 &0 &3 &0 &0 &0 &-1
\\ \noalign{\vspace*{2pt}}
-1 &0 &0 &0 &3 &0 &-1 &-1
\\ \noalign{\vspace*{2pt}}
0 &-1 &0 &0 &0 &3 &-1 &-1
\\ \noalign{\vspace*{2pt}}
0 &0 &-1 &0 &-1 &-1 &3 &0
\\ \noalign{\vspace*{2pt}}
0 &0 &0 &-1 &-1 &-1 &0 &3
\end{array}
\right]
\end{equation*}
which has PST at time $\pi /2$ for the pairs $(1,6), (2,5), (3,8),
(4,7)$. It turns out that $L_{1}+L_{2}$ has PST at time $\pi /2$ between
the pairs $(1,3), (2,4), (5,7), (6,8)$. From the above collection of
cases, we find for example that if $w_{1}\in \mathbb{Q}$ and
${{w_{2}}} \in \overline{\mathbb{Q}}$ (or $w_{1}\in \overline{\mathbb{Q}}$
and ${{w_{2}}}\in \mathbb{Q}$), then $G_{1}
 \tensor[_{w_{1}}]{\odot }{_{{{w_{2}}}}} G_{2}$ has the intriguing property
that there is PGST between the pairs $(1,8), (1,11), (1,14)$ (among
others).
\end{example}

\section{Optimality}\label{optim}

\subsection{Timing errors}

In \cite{GKLPZ}, the authors analyse the sensitivity of the
probability of state transfer in the presence of small perturbations.
Bounds on the probability of state transfer with respect to timing
errors and with respect to manufacturing errors were given in the most
general setting where no information is known about the graph in
question. Specifically, suppose that under $XXX$ dynamics, a graph
$G$ on $n$ vertices has PST from vertex $1$ to vertex $2$ at time
$t_{0}$. Suppose further that there is a small perturbation so that the
readout time is instead $t_{0}+h$, where $|h| < \frac{\pi }{\lambda
_{n}}$, and where $\lambda_{n}$ is the largest eigenvalue of the
corresponding Laplacian. Decompose the Laplacian as $L=Q\Lambda Q^{T}$,
where $\Lambda =\diag (\lambda_{1}=0, \lambda_{2}, \dots , \lambda
_{n})$, with $0< \lambda_{2}\leq \dots \leq \lambda_{n}$, and $Q$ is an
orthogonal matrix of corresponding eigenvectors. If~$\mathbf{q}_{1}$ and
$\mathbf{q}_{2}$ are the first and second columns of $Q^{T}$,
respectively, then for some $\theta \in \mathbb{R}$ we have
$e^{i\theta }\mathbf{q}_{1}=e^{it_{0}\Lambda }\mathbf{q}_{2}$. Setting
$M = \diag (e^{ih\lambda_{1}}, \dots , e^{ih\lambda_{n}}) e^{i\theta
}$, it follows that
\begin{eqnarray*}
p(t_{0})-p(t_{0}+h)
&=&
1-|\mathbf{q}_{1}^{T}M\mathbf{q}_{1}|^{2} .
\end{eqnarray*}
In the special case that $G$ is Hadamard diagonalizable, we have
$Q=\frac{1}{\sqrt{n}}H$, where $H$ is a Hadamard matrix, so we can say
more. In that case,
%
\begin{eqnarray}
|\mathbf{q}_{1}^{T}M\mathbf{q}_{1}|
&=&\frac{1}{n}\left| \sum_{j=1}
^{n} e^{ih\lambda_{j}}\right| .
\end{eqnarray}
This suggests that, in order to find a lower bound for $|\mathbf{q}
_{1}^{T}M\mathbf{q}_{1}|$ (and thus an upper bound for $p(t_{0})-p(t
_{0}+h)$), the goal should be to make the numbers $e^{ih\lambda_{j}}$
as closely-spaced on the complex unit circle as possible. This agrees
with the known fact that minimizing the spectral spread has the
effecting of maximizing the bound for the fidelity of state transfer due
to timing errors \cite{Kay06}. Thus, this remark is not surprising
but rather confirms the known rule while at the same time providing a
more accurate bound on timing errors for Hadamard diagonalizable graphs.

\subsection{Manufacturing errors: sparsity of graphs with PST}

It is desirable to minimize the number of edges that need to be
engineered in a graph (so as to minimize manufacturing errors), so one
question of interest in the theory of perfect state transfer is how
sparse a graph with perfect state transfer can be. Among the sparsest
known graph with PST is the $k$-cube, which has $2^{k}$ vertices, each
with degree $k$. We now show that if we restrict our attention to
Hadamard diagonalizable unweighted graphs, then for $k \le 4$ the
$k$-cube is indeed the sparsest connected graph with PST.\looseness=1

\begin{theorem}
\label{thm:hypercube_sparse}
Let $G$ be a simple, connected, unweighted $r$-regular graph on $n$
vertices. Suppose further that $G$ is Hadamard diagonalizable, has
perfect state transfer at $\pi /2$, and that $r \leq 4$. Then
$n \leq 2^{r}$.
\end{theorem}
\begin{proof}
If $L$ is the Laplacian of $G$ then the result follows by computing some
quantities of the form $\mathrm{Tr}(L^{k})$ ($k \geq 0$ is an integer)
in two different ways. First, if $\lambda_{1},\ldots ,\lambda_{n}$ are
the eigenvalues of $L$ then $\mathrm{Tr}(L^{k}) = \sum_{j=1}^{n}
\lambda_{j}^{k}$, and we know by the Gershgorin circle theorem that
$0 \leq \lambda_{j} \leq 2r$ for each $j$. On the other hand,
$L = rI - A$, where $A$ is the adjacency matrix of $G$, so $
\mathrm{Tr}(L) = rn - \mathrm{Tr}(A)$ and $\mathrm{Tr}(L^{2}) = r^{2}n
- 2r\mathrm{Tr}(A) + \mathrm{Tr}(A^{2})$. Since $A$ is simple, we know
that $\mathrm{Tr}(A) = 0$ and it is straightforward to compute
$\mathrm{Tr}(A^{2}) = rn$. Thus we have the following system of
equations:
\begin{equation*}
\sum_{j=1}^{n} \lambda_{j} = rn \quad \text{and} \quad \sum_{j=1}^{n}
\lambda_{j}^{2} = rn(r + 1).
\end{equation*}

If we let $c_{\lambda }$ denote the number of eigenvalues of $L$ equal
to $\lambda $ (with the convention that if $\lambda $ is not an
eigenvalue, then $c_{\lambda }=0$), then these equations tell us that
%
\begin{eqnarray}
\label{eq:sparse}
\sum_{j=1}^{r} (2j)c_{2j} = rn \quad \text{and} \quad \sum_{j=1}^{r}(2j)^{2}c
_{2j} = rn(r + 1).
\end{eqnarray}
If we add in the equation $\sum_{j=1}^{r} c_{2j} = n-1$ (since one of
the eigenvalues equals $0$), then we have a system of $3$ linear
equations in the variables $n,c_{2},c_{4},\ldots ,c_{2r}$. If
$r \leq 2$ then it is straightforward to solve this system of equations
to get $n = 2^{r}$. If $r = 3$ then by adding the equation $c_{2} + c
_{6} = c_{4} + 1$ (since we know that half of $L$'s eigenvalues must
belong to each equivalence class mod $4$) we can similarly solve the
system of equations to get $n = 8 = 2^{r}$.

For the $r = 4$ case, we use the Equations~{\eqref{eq:sparse}} together
with the equation $c_{2} + c_{6} = c_{4} + c_{8} + 1$ (again, because
the eigenvalues are split evenly between the mod $4$ equivalence
classes). These equations together can be reduced to the system of
equations $c_{2} = 3n/8 -\nobreak  2$, $c_{4} = 3n/8$, $c_{6} = n/8 + 2$, and
$c_{8} = n/8 - 1$. To reduce this system further and get a unique
solution, we need to compute $\mathrm{Tr}(L^{3})$ in two different ways
(similar to at the start of the proof): $\mathrm{Tr}(L^{3}) = \sum
_{j=1}^{n} \lambda_{j}^{3} = r^{3}n - 3r^{2}\mathrm{Tr}(A) + 3r
\mathrm{Tr}(A^{2}) - \mathrm{Tr}(A^{3}) = r^{3}n + 3r^{2}n -
\mathrm{Tr}(A^{3})$. Since $\mathrm{Tr}(A^{3}) \ge 0$ we arrive at the
inequality $\sum_{j=1}^{n} \lambda_{j}^{3} \leq r^{2}n(r+3)$, which is
equivalent to $\sum_{j=1}^{r}(2j)^{3}c_{2j} \leq r^{2}n(r+3)$. Plugging
in $r = 4$ then gives
\begin{equation*}
8c_{2} + 64c_{4} + 216c_{6} + 512c_{8} \leq 112n.
\end{equation*}
It is then straightforward to substitute the equations $c_{2} = 3n/8 -
2$, $c_{4} = 3n/8$, $c_{6} = n/8 + 2$, and $c_{8} = n/8 - 1$ into this
inequality to get $n \leq 2^{r} = 16$, as desired.
\end{proof}

It seems reasonable to believe that {Theorem~\ref{thm:hypercube_sparse}}
could be generalized to arbitrary $r$, but the method of proof that we
used does not seem to generalize in a straightforward way, as there are
no more obvious equations or inequalities involving the $c_{2j}$'s that
we can use. For example, if we try to extend the proof of
{Theorem~\ref{thm:hypercube_sparse}} to the $r = 5$ case, we might try
computing $\mathrm{Tr}(L^{4})$ in two different ways. However, we then
end up with an equation involving both $-\mathrm{Tr}(A^{3})$ and
$+\mathrm{Tr}(A^{4})$, and it is not clear how to bound such a quantity.

\section{Acknowledgements}
The authors are grateful to an anonymous referee, whose constructive
comments resulted in improvements to the paper. N.J., S.K., and S.P.
were supported by {NSERC} Discovery Grant number RGPIN-2016-04003, RGPIN/6123-2014, and 1174582, respectively;
N.J.\ was also supported by
a  {Mount Allison Marjorie Young Bell Faculty Fund}; R.S.\ was supported
through a {NSERC} Undergraduate Student Research Award; X.Z.\ was
supported by the {University of Manitoba's Faculty of Science and Faculty of Graduate Studies}.





\end{document}